\documentclass{lmcs}
\pdfoutput=1

\usepackage{lastpage}
\lmcsdoi{16}{4}{16}
\lmcsheading{}{\pageref{LastPage}}{}{}%
{Sep.~17,~2019}{Dec.~15,~2020}{}


\usepackage{breakurl}

\usepackage[active]{srcltx}
\usepackage[utf8]{inputenc}
\usepackage{mathpartir}
\usepackage{latexsym}
\usepackage{amsmath}
\usepackage{amssymb}
\usepackage{amscd}
\usepackage{macros}
\usepackage{color}
\usepackage{graphics}
\usepackage{version}
\usepackage{soul}
\usepackage{placeins}
\usepackage{stmaryrd}
\usepackage[normalem]{ulem}

\usepackage{psi-macros}
\usepackage{microtype}

\theoremstyle{plain}

\newcommand{\provtrans}[4]{#1 \; \xrightarrow[#3]{#2} \; #4}
\newcommand{\framedprovtrans}[5]{#1 \frames \provtrans{#2}{#3}{#4}{#5}}

\begin{document}

\title[Psi-Calculi Revisited: Connectivity and Compositionality]
      {Psi-Calculi Revisited:\texorpdfstring{\\}{} Connectivity and Compositionality}
\author{Johannes {\AA}man Pohjola}
\address{Data61/CSIRO, Sydney, Australia
  \\ University of New South Wales, Sydney, Australia}
\email{johannes.amanpohjola@data61.csiro.au}

\begin{abstract}
Psi-calculi is a parametric framework for process calculi similar to popular pi-calculus extensions such as the explicit fusion calculus, the applied pi-calculus and the spi calculus.
Mechanised proofs of standard algebraic and congruence properties of bisimilarity apply to all calculi within the framework.

A limitation of psi-calculi is that communication channels must be symmetric and transitive.
In this paper, we give a new operational semantics to psi-calculi that allows us to lift these
restrictions and simplify some of the proofs.
The key technical innovation is to annotate transitions with a
\emph{provenance}---a description of the scope and channel they originate from.

We give mechanised proofs that our extension is conservative,
and that the standard algebraic and
congruence properties of strong and weak bisimilarity are maintained.
We show correspondence with a reduction semantics and barbed bisimulation.
We show how a pi-calculus with preorders that was previously beyond the scope of psi-calculi
can be captured,
and how to encode mixed choice under very strong quality criteria.
\end{abstract}

\maketitle

\section{Introduction}\label{sec:intro}

This paper is mainly concerned with \emph{channel connectivity}, by which we mean the relationship
that describes which input channels are connected to which output channels in a setting with message-passing concurrency.
In the {\pic}~\cite{DBLP:journals/iandc/MilnerPW92a},
channel connectivity is syntactic identity: in the process
\[\underline{a}(x).P \parop \overline{b}\,y.Q\]
\noindent where one parallel component is waiting to receive on channel $a$ and the other is waiting to send on channel $b$, interaction is possible only if $a=b$.

Variants of the {\pic} may have more interesting channel connnectivity.
The explicit fusion calculus pi-F~\cite{gardner.wischik:explicit-fusions-mfcs} extends
the {\pic} with a primitive for \emph{fusing} names;
once fused, they are treated as being for all purposes one and the same.
Channel connectivity is then given by the equivalence closure of the name fusions.
For example, if we extend the above example with the fusion $(a=b)$

\[\underline{a}(x).P \parop \overline{b}\,y.Q \parop (a=b)\]

\noindent then communication is possible.
Other examples may be found in e.g.~calculi for wireless communication~\cite{DBLP:journals/tcs/NanzH06}, where channel connectivity can be used to directly model the network's topology.

Psi-calculi~\cite{bengtson.johansson.ea:psi-calculi-long} is a family of applied process calculi,
where standard meta-theoretical results,
such as the algebraic laws and congruence properties of bisimulation,
have been established once and for all through mechanised proofs~\cite{DBLP:journals/jar/BengtsonPW16} for all members of the family.
Psi-calculi generalises e.g.~the {\pic} and the explicit fusion calculus in several ways.
In place of atomic names it allows
channels and messages to be taken from an (almost) freely chosen term language.
In place of fusions, it admits the formulas of an (almost) freely chosen logic as first-class
processes. Channel connectivity is determined by judgements of said logic, with one restriction:
the connectivity thus induced must be symmetric and transitive.


The main contribution of the present paper is a new way to define the semantics of psi-calculi that
lets us lift this restriction, without sacrificing any of the algebraic laws
and compositionality properties.
It is worth noting that this was previously believed to be impossible:
Bengtson et al.~\cite[p.~14]{bengtson.johansson.ea:psi-calculi-long} even offer counterexamples
to the effect that without symmetry and transitivity, scope extension is unsound.
However, a close reading reveals that these counterexamples apply only to their particular choice of
labelled semantics, and do not rule out the possibility that the counterexamples
could be invalidated by a rephrasing of the labelled semantics such as ours.

The price we pay for this increased generality is more complicated transition labels:
we decorate input and output labels with a \emph{provenance} that keeps track of which prefix a
transition originates from. The idea is that if I am an input label and you are an output label,
we can communicate if my subject is your provenance, and vice versa.
This is offset by other simplifications of the semantics and associated
proofs that provenances enable.

\paragraph{Contributions} This paper makes the following specific technical contributions:

\begin{itemize}
  \item We define a new semantics of psi-calculi that lifts the requirement that
    channel connectivity must be symmetric and transitive,
    using the novel technical device of provenances.
    (Section~\ref{sec:definitions})
  \item We prove that strong and weak bisimulation is a congruence and satisfies the usual
    algebraic laws such as scope extension.
    Interestingly, provenances can be ignored for the purpose of bisimulation.
    These proofs are machine-checked in Nominal Isabelle~\cite{U07:NominalTechniquesInIsabelleHOL}.
    (Section~\ref{sec:bisimulation})
  \item We prove, again using Nominal Isabelle, that this paper's developments
    constitute a conservative extension of the original psi-calculi.
    (Section~\ref{sec:validation})
  \item To further validate our semantics, we define
    a reduction semantics and strong barbed congruence,
    and show that they agree with their labelled counterparts.
    (Section~\ref{sec:validation})    
  \item We capture a pi-calculus with preorders
    by Hirschkoff et al.~\cite{DBLP:conf/lics/HirschkoffMS13},
    that was previously beyond the scope of psi-calculi because of its
    non-transitive channel connectivity.
    The bisimilarity we obtain turns out to coincide with that of Hirschkoff et al.
    (Section~\ref{sec:prepi})
  \item 
    We exploit non-transitive connectivity to show that
    mixed choice is a derived operator of psi-calculi in a very strong sense:
    its encoding is fully abstract and satisfies strong operational correspondence.
    (Section~\ref{sec:choice})
\end{itemize}

\noindent
This paper is an extended version of \cite{DBLP:conf/forte/Pohjola19}.
In this version, we have extended many of the meta-theoretical results
and the associated Isabelle formalisation from strong to weak bisimulation
(Section~\ref{sec:weakbisimulation}).
We have also added a discussion of the aforementioned
counterexamples by Bengtson et al.~\cite{bengtson.johansson.ea:psi-calculi-long}
(Section~\ref{sec:counterexamples}), and a more thorough motivation of
the design decisions we have made when introducing provenances (Section~\ref{sec:design}).
Moreover, the other sections have been edited for detail and clarity.
We have opted against including the full proofs;
the interested reader is referred to the associated technical report~\cite{pohjola:newpsireport}
and Isabelle formalisation.
Isabelle proofs are available online.%
\footnote{\url{https://github.com/IlmariReissumies/newpsi}}

\section{Definitions}\label{sec:definitions}

This section introduces core definitions such as syntax and semantics.
Many definitions are shared with the original presentation of psi-calculi, so this section
also functions as a recapitulation of \cite{bengtson.johansson.ea:psi-calculi-long}. We will highlight
the places where the two differ.
For readers who desire a gentler introduction to psi-calculi than the present
paper offers, we recommend \cite{Johansson10}.


Psi-calculi is built on the theory of nominal sets~\cite{Gabbay01anew}, which allows us to reason
formally up to alpha-equivalence without committing to any particular syntax of the term
language.
We assume a countable set of \emph{names} $\nameset$ ranged over by $a,b,c,\dots,x,y,z$. 
A \emph{nominal set} is a set equipped with a permutation action $\cdot$;
intuitively, if $X \in \mathbf{X}$ and $\mathbf{X}$ is a nominal set, then $(x\;y)\cdot X$, which denotes
$X$ with all occurrences of the name $x$ swapped for $y$ and vice versa, is also an element of
$\mathbf{X}$.
$\supp{X}$ (the \emph{support} of $X$) is, intuitively,
the set of names such that swapping them changes $X$.
We write $a \freshin X$ (``$a$ is fresh in $X$) for $a \notin \supp{X}$.
We overload $\freshin$ to sequences of names: $\vec{a} \freshin X$ means that for
each $a_i$ in $\vec{a}$, $a_i \freshin X$.
Similarly, $a \freshin X,Y$ abbreviates $a \freshin X \wedge a \freshin Y$.
A nominal set $\mathbf{X}$ has \emph{finite support} if for every $X\in\mathbf{X}$,
$\supp{X}$ is finite.
A function symbol $f$ is \emph{equivariant} if $p\cdot f(x) = f(p\cdot x)$,
where $p$ is a sequence of permutations $(x\;y)$;
this generalises to
$n$-ary function symbols in the obvious way.
Whenever we define inductive syntax with names, it is implicitly quotiented by permutation of
bound names, so e.g.~$(\nu x)\opa{a}{x} = (\nu y)\opa{a}{y}$ if $x,y \freshin a$.

Psi-calculi is parameterised on an arbitrary term language and a logic of environmental assertions:

\begin{defi}[Parameters]
A \emph{psi-calculus} is a 7-tuple $(\terms,\assertions,\conditions,\vdash,\otimes,\unit,\chcon)$ with
three finitely supported nominal sets:

\begin{enumerate}
  \item $\terms$, the \emph{terms}, ranged over by $M,N,K,L,T$;
  \item $\assertions$, the \emph{assertions}, ranged over by $\Psi$; and
  \item $\conditions$, the \emph{conditions}, ranged over by $\varphi$.
\end{enumerate}

\noindent We assume each of the above is equipped with a substitution function
\subst{\_}{\_} that substitutes (sequences of) terms for names.
The remaining three parameters are equivariant function symbols written in infix:
\begin{center}
\begin{tabular}{ll}
   $\vdash\; : \assertions \times \conditions \Rightarrow \kwd{bool}$ & (entailment) \\
   $\otimes : \assertions \times \assertions \Rightarrow \assertions$ & (composition) \\
   $\unit : \assertions$ & (unit) \\
   $\chcon\;: \terms \times \terms \Rightarrow \conditions$ & (channel connectivity)
\end{tabular}
\end{center}
\end{defi}

\noindent Intuitively, $M \chcon K$ means the prefix $M$ can send a message to the prefix $K$.
The substitution functions must satisfy certain natural criteria wrt.~their treatment of names;
see \cite{bengtson.johansson.ea:psi-calculi-long} for the details.
We use $\sigma$ to range over substitutions $\subst{\ve{x}}{\ve{T}}$. A substitution
$\subst{\ve{x}}{\ve{T}}$ is \emph{well-formed} if $|\ve{x}| = |\ve{T}|$ and $\ve{x}$ are
pairwise distinct. Unless otherwise specified, we only consider well-formed substitutions.

\begin{defi}[Static equivalence]
  Two assertions $\Psi,\Psi'$ are \emph{statically equivalent}, written
  $\Psi \simeq \Psi'$, if
  $\forall \varphi.\; \Psi \vdash \varphi \; \Leftrightarrow \; \Psi' \vdash \varphi$.
\end{defi}

\begin{defi}[Valid parameters]
A psi-calculus is \emph{valid} if $(\assertions/{\simeq},\otimes,\unit)$ form an abelian monoid.
\end{defi}

Note that since the abelian monoid is closed, static equivalence is preserved by composition.
Henceforth we will only consider valid psi-calculi.
The original presentation of psi-calculi had $\sch$ for channel equivalence in place of
our $\chcon$, and required that channel equivalence be symmetric
(formally, $\Psi \vdash M \sch K$ iff $\Psi \vdash K \sch M$)
and transitive. 

\begin{defi}[Process syntax]
  The \emph{processes} (or \emph{agents}) $\processes$, ranged over by $P,Q,R$, are inductively defined by the grammar

  \begin{center}
  \begin{tabular}{rcll}
    $P$ & :=     & $\nil$ & (nil) \\
      &  & $\pass{\Psi}$ & (assertion) \\
      &  & $\outprefix{M}{N}.P$ & (output) \\
      &  & $\inprefix{M}{\vec{x}}{N}.P$ & (input) \\
      &  & ${\caseonly\;{\ci{\ve{\varphi}}{\ve{P}}}}$ & (case) \\
      &  & $P \parop Q$ & (parallel composition) \\
      &  & $(\nu x)P$ & (restriction) \\
      &  & $!P$ & (replication)
  \end{tabular}
  \end{center}

\noindent A process is \emph{assertion-guarded}
if all assertions occur underneath an input or output prefix.
We require that in $!P$, $P$ is guarded;
that in ${\caseonly\;{\ci{\ve{\varphi}}{\ve{P}}}}$, all $\ve{P}$ are guarded;
and that in $\inprefix{M}{\vec{x}}{N}\sdot P$ it holds that $\vec{x} \subseteq \supp{N}$.
We will use $P_G,Q_G$ to range over assertion-guarded processes.
A process $P$ is \emph{prefix-guarded} if its outermost operator is an input
or output prefix.
\end{defi}

Restriction, replication and parallel composition are standard.
We lift restriction to sequences of names by letting $(\nu\ve{a})P$
abbreviate $(\nu a_0)(\nu a_1)\dots(\nu a_i)P$; in
particular, $(\nu \epsilon)P = P$.
$\outprefix{M}{N}.P$ is a process ready to send
the message $N$ on channel $M$, and then continue as $P$.
Similarly, $\inprefix{M}{\vec{x}}{N}.P$ is a process ready to receive
a message on channel $M$ that matches the pattern $(\lambda\vec{x})N$.
We sometimes write $\outprefix{M}{N}$ to stand for $\outprefix{M}{N}.0$,
and similarly for input. We elide the object $N$ when it is unimportant.

The process $\pass{\Psi}$ asserts a fact $\Psi$ about the environment.
Intuitively, $\pass{\Psi} \parop P$ means that $P$ executes in an environment
where all conditions entailed by $\Psi$ hold. $P$ may itself
contain assertions that add or retract conditions.
Environments can evolve dynamically: as a process reduces,
assertions may become unguarded and thus added to the
environment.
${\caseonly\;{\ci{\ve{\varphi}}{\ve{P}}}}$ is a process that may act as
any $P_i$ whose guard $\varphi_i$ is entailed by the environment.
For a discussion of why replication and case must be assertion-guarded we refer to
\cite{bengtson.johansson.ea:psi-calculi-long,DBLP:conf/lics/JohanssonBPV10}.
We use $\casesep$ to denote composition of $\caseonly$ statements,
so e.g.~%
$\caseonly\;{\ci{\varphi}{P}}\casesep{\ci{\varphi'}{Q}}$
desugars to ${\caseonly\;\ci{\varphi,\varphi'}{P,Q}}$.

The assertion environment of a process is described by its \emph{frame}:

\begin{defi}[Frames]\label{def:frame}
  The \emph{frame} of $P$, written $\frameof{P} = \framepair{\frnames{P}}{\frass{P}}$
  where $\frnames{P}$ bind into $\frass{P}$, is defined as
  \[
  \begin{array}{rlll}
    \frameof{\pass{\Psi}} & = & \framepair{\epsilon}{\Psi} &
    \\
    \frameof{P \parop Q} & = & \frameof{P} \otimes \frameof{Q} &
    \\
     \frameof{(\nu x)P} & = & (\nu x)\frameof{P} &
     \\
     \frameof{P} & = & \unit & \mbox{otherwise}
  \end{array}
  \]
%
  \noindent where name-binding and composition of frames is defined as
  $(\nu x)\framepair{\frnames{P}}{\frass{P}} = \framepair{x,\frnames{P}}{\frass{P}}$,
  and, if $\frnames{P} \freshin \frnames{Q},\frass{Q}$ and $\frnames{Q} \freshin \frass{P}$,
  \[\framepair{\frnames{P}}{\frass{P}} \otimes \framepair{\frnames{Q}}{\frass{Q}} = \framepair{\frnames{P},\frnames{Q}}{(\frass{P} \otimes \frass{Q})}\]

\noindent where $\nu$ binds stronger than $\otimes$.
We overload $\Psi$ to denote the frame $(\nu \epsilon)\Psi$.
\end{defi}

We extend entailment to frames as follows: $\frameof{P} \vdash \varphi$ holds if,
for some $\frnames{P},\frass{P}$ such that $\frameof{P} = \framepair{\frnames{P}}{\frass{P}}$
and $\frnames{P} \freshin \varphi$,
$\frass{P} \vdash \varphi$.
The freshness side-condition $\frnames{P} \freshin \varphi$ is important because it allows
assertions to be used for representing local state. By default, the assertion environment is
effectively a form of global non-monotonic state,
which is not always appropriate for modelling distributed processes.
With $\nu$-binding we can recover locality by writing e.g.~$(\nu x)(\pass{x=M} \parop P)$
for a process $P$ with a local variable $x$.

The notion of \emph{provenance} is the main novelty of our semantics.
It is the key technical device used to make our semantics compositional:

\begin{defi}[Provenances]
The \emph{provenances} $\provenances$, ranged over by $\pi$, are either $\bot$ or of form $(\nu\ve{x};\ve{y})M$, where $M$ is a term, and $\ve{x},\ve{y}$ bind into $M$.
\end{defi}

We write $M$ for $(\nu\epsilon;\epsilon)M$. When $\ve{x},\ve{y} \freshin \ve{x'},\ve{y'}$ and $\ve{x} \freshin \ve{y}$, we interpret the expression $(\nu\ve{x};\ve{y})(\nu\ve{x'};\ve{y'})M$ as $(\nu\ve{x}\,\ve{x'};\ve{y}\,\ve{y'})M$. Furthermore, we equate $(\nu\ve{x};\ve{y})\bot$ and $\bot$.
Let $\pi\downarrow$ denote the result of moving all binders from the outermost binding sequence to the innermost; that is,
$(\nu\ve{x};\ve{y})M{\downarrow} = (\nu\epsilon;\ve{x},\ve{y})M$.
Similarly, $\pi\downarrow\ve{z}$ denotes the result of inserting $\ve{z}$ at the end of the outermost binding sequence:
formally, $(\nu\ve{x};\ve{y})M\downarrow\ve{z} = (\nu\ve{x},\ve{z};\ve{y})M$.

Intuitively, a provenance describes the origin of an input or output transition.
For example, if an output transition is annotated with $(\nu\ve{x};\ve{y})M$, the sender
is an output prefix with subject $M$ that occurs underneath the $\nu$-binders $\ve{x},\ve{y}$.
For technical reasons, these binders are partitioned into two distinct sequences.
The intention is that $\ve{x}$ are the frame binders,
while $\ve{y}$ contains binders that occur underneath case and replication;
these are not part of the frame, but may nonetheless bind into $M$.
We prefer to keep them separate because the $\ve{x}$ binders
are used for deriving $\vdash$ judgements,
but $\ve{y}$ are not (cf.~Definition~\ref{def:frame}).

\begin{defi}[Labels]
The \emph{labels} $\labels$, ranged over by $\alpha,\beta$, are:
\begin{center}
\begin{tabular}{ll}
    $\bout{\ve{x}}{M}{N}$ & (output) \\
    $\inlabel{M}{N}$ & (input) \\
    $\tau$ & (silent)
\end{tabular}
\end{center}
\noindent The bound names of $\alpha$, written $\bn{\alpha}$, is $\ve{x}$ if $\alpha = \bout{\ve{x}}{M}{N}$ and $\epsilon$ otherwise.
The subject of $\alpha$, written $\subj{\alpha}$, is $M$ if $\alpha = \bout{\ve{x}}{M}{N}$ or $\alpha = \inlabel{M}{N}$.
Analogously, the object of $\alpha$, written $\obj{\alpha}$,
is $N$ if $\alpha = \bout{\ve{x}}{M}{N}$ or $\alpha = \inlabel{M}{N}$.
\end{defi}

While the provenance describes the origin of a transition, a label describes how it can
interact.
For example, a transition labelled with $\inlabel{M}{N}$ indicates readiness to receive a
message $N$ from an output prefix with subject $M$.

\begin{defi}[Operational semantics]
The transition relation $\mathord{\apitransarrow{}}\subseteq \assertions \times \processes \times \labels \times \provenances \times \processes$ is inductively defined by the rules in Table~\ref{table:full-struct-free-labeled-operational-semantics}.
We write $\framedprovtrans{\Psi}{P}{\alpha}{\pi}{P'}$ for $(\Psi,P,\alpha,\pi,P')\in\mathord{\apitransarrow{}}$. In transitions, $\bn{\alpha}$ binds into $\obj{\alpha}$ and $P'$.
\end{defi}

\begin{table*}[tb]

\begin{mathpar}

\inferrule*[Left=\textsc{In}]
    {\Psi \vdash K \chcon M}   
{\framedprovtrans{\Psi}{\inprefix{M}{\ve{y}}{N} \sdot
P}{\inlabel{K}{N}\subst{\ve{y}}{\ve{L}}}{M}{P\subst{\ve{y}}{\ve{L}}}}


\inferrule*[left=\textsc{Out}]
    {\Psi \vdash M \chcon K}
    {\framedprovtrans{\Psi}{\outprefix{M}{N} \sdot P}{\outlabel{K}{N}}{M}{P}}


\inferrule*[left=\textsc{ParL},  right={$\inferrule{}{\bn{\alpha} \freshin Q}
$}]
{\framedprovtrans{\frass{Q} \ftimes \Psi}{P} {\alpha}{\pi}{P'}}
{\framedprovtrans{\Psi}{P \pll Q}{\alpha}{\pi\downarrow\frnames{Q}}{P' \pll Q}}

\inferrule*[left=\textsc{ParR},  right={$\inferrule{}{\bn{\alpha} \freshin P}
$}]
{\framedprovtrans{\frass{P} \ftimes \Psi}{Q} {\alpha}{\pi}{Q'}}
{\framedprovtrans{\Psi}{P \pll Q}{\alpha}{(\nu\frnames{P})\pi}{P \pll Q'}}

\inferrule*[left=\textsc{Com}, right={$\ve{a} \freshin Q$}]
 {
 \framedprovtrans{\frass{Q} \ftimes \Psi}{P}{\outlabel{M}{(\nu \ve{a})N}}{(\nu\frnames{P};\ve{x})K}{P'} \\
  \framedprovtrans{\frass{P} \ftimes \Psi}{Q}{\inlabel{K}{N}}{(\nu\frnames{Q};\ve{y})M}{Q'} 
  }
       {\framedprovtrans{\Psi}{P \pll Q}{\tau}{\bot}{(\nu \ve{a})(P' \pll Q')}}

\inferrule*[left={\textsc{Case}}]
    {\framedprovtrans{\Psi}{P_i}{\alpha}{\pi}{P'} \\ \Psi \vdash \varphi_i}
    {\framedprovtrans{\Psi}{\case{\ci{\ve{\varphi}}{\ve{P}}}}{\alpha}{\pi\downarrow}{P'}}

\inferrule*[left=\textsc{Scope}, right={$b \freshin \alpha,\Psi$}]
    {\framedprovtrans{\Psi}{P}{\alpha}{\pi}{P'}}
    {\framedprovtrans{\Psi}{(\nu b)P}{\alpha}{(\nu b)\pi}{(\nu b)P'}}
    
\inferrule*[left=\textsc{Open}, right={$\inferrule{}{b \freshin \ve{a},\Psi,M\\\\b \in \names{N}}$}]
    {\framedprovtrans{\Psi}{P}{\outlabel{M}{(\nu \ve{a})N}}{\pi}{P'}}
    {\framedprovtrans{\Psi}{(\nu b)P}{\outlabel{M}{(\nu \ve{a} \cup \{b\})N}}{(\nu b)\pi}{P'}}

\inferrule*[left=\textsc{Rep}]
   {\framedprovtrans{\Psi}{P \pll !P}{\alpha}{\pi}{P'}}
   {\framedprovtrans{\Psi}{!P}{\alpha}{\pi\downarrow}{P'}}

\end{mathpar}
\caption{\rm Structured operational semantics. A symmetric version of \textsc{Com} is elided. In the rule $\textsc{Com}$ we assume that
  $\frameof{P} = \framepair{\frnames{P}}{\Psi_P}$ and
  $\frameof{Q} = \framepair{\frnames{Q}}{\Psi_Q}$ where
  $\frnames{P}$ is fresh for $\Psi$ and $Q$,
  $\ve{x}$ is fresh for $\Psi, \frass{Q}, P$, and
  $\frnames{Q},\ve{y}$ are similarly fresh.
In rule
\textsc{ParL} we assume that $\frameof{Q} = \framepair{\frnames{Q}}{\Psi_Q}$
where $\frnames{Q}$ is fresh for
$\Psi, P, \pi$ and $\alpha$.
\textsc{ParR} has the same freshness conditions but with the roles of $P,Q$ swapped.
In $\textsc{Open}$ the expression $\tilde{a} \cup \{b\}$ means the sequence
$\tilde{a}$ with $b$ inserted anywhere.
}
\label{table:full-struct-free-labeled-operational-semantics}
\end{table*}

Note that to avoid clutter, the freshness conditions for some of the rules are stated in the caption of Table~\ref{table:full-struct-free-labeled-operational-semantics}.

The operational semantics differs from \cite{bengtson.johansson.ea:psi-calculi-long}
mainly by the inclusion of provenances: anything underneath the transition arrows is novel.

The \textsc{Out} rule states that in an environment where $M$ is connected to $K$,
the prefix $\outprefix{M}{N}$ may send a message $N$ from $M$ to $K$.
The \textsc{In} rule is dual to \textsc{Out}, but also features pattern-matching.
If the message is an instance of the pattern, as witnessed by a substitution,
that substitution is applied to the continuation $P$.

In the \textsc{Com} rule, we see how provenances are used to determine when two
processes can interact.
Specifically, a communication between $P$ and $Q$ can be derived if $P$ can send a message to $M$ using prefix $K$,
and if $Q$ can receive a message from $K$ using prefix $M$.
Because names occurring in $M$ and $K$ may be local to $P$ and $Q$ respectively,
we must be careful not to conflate the local names of one with the other;
this is why the provenance records all binding names that occur above $M,K$ in the process syntax.
Note that even though we identify frames and provenances up to alpha,
the \textsc{Com} rule uses syntactically identical binding sequences $\ve{b}_P,\ve{b}_Q$ in two roles:
as frame binders $\frameof{P} = \framepair{\frnames{P}}{\frass{P}}$,
and as the outermost provenance binders.
By thus insisting that these binding sequences are chosen to coincide,
we ensure that the $K$ on $Q$'s label
really is the same as the $K$ in $P$'s provenance.

It is instructive to compare our \textsc{Com} rule with the original:

\begin{mathpar}
\inferrule*[left=\textsc{Com-Old}, right={$\inferrule{}{\ve{a} \freshin Q }$}]
 {\framedtrans{\frass{Q} \ftimes \Psi}{P}{\bout{\ve{a}}{M}{N}}{P'} \\
  \framedtrans{\frass{P} \ftimes \Psi}{Q}{\inlabel{K}{N}}{Q'} \\
  \Psi \ftimes \frass{P} \ftimes \frass{Q} \vdash M \sch K
  }
       {\framedtrans{\Psi}{P \pll Q}{\tau}{(\nu \ve{a})(P' \pll Q')}}
\end{mathpar}

\noindent where $\frameof{P} =\framepair{\frnames{P}}{\frass{P}}$ and
$\frameof{Q} = \framepair{\frnames{Q}}{\frass{Q}}$ and $\frnames{P} \freshin \Psi, \frnames{Q}, Q, M, P$ and $\frnames{Q} \freshin \Psi, \frnames{Q}, Q, K, P$.
Here we have no way of knowing if $M$ and $K$ are able
to synchronise other than making a channel equivalence judgement.
Hence any derivation involving \textsc{Com-Old} makes three channel equivalence judgements:
once each in \textsc{In}, \textsc{Out} and \textsc{Com-Old}. With \textsc{Com} we only make one ---
or more accurately, we make the exact same judgement twice, in \textsc{In} resp.~\textsc{Out}.
Eliminating the redundant judgements is crucial: the reason \textsc{Com-Old} needs associativity
and commutativity is to stitch these three judgements together, particularly
when one end of a communication is swapped for a bisimilar process
that allows the same interaction via different prefixes.

Note also that \textsc{Com} has fewer freshness side-conditions.
A particularly unintuitive aspect of \textsc{Com-Old} is that it
requires $\frnames{P} \freshin M$ and $\frnames{Q} \freshin K$, but not
 $\frnames{P} \freshin K$ and $\frnames{Q} \freshin M$: we would expect that all bound
names can be chosen to be distinct from all free names, but adding the missing
freshness conditions makes scope extension unsound~\cite[pp.~56-57]{Johansson10}.
With \textsc{Com}, it becomes clear why:
because $\frnames{Q}$ binds into $M$.

All the other rules can fire independently of what the provenance of the premise is.
They manipulate the provenance, but only for bookkeeping purposes:
in order for the \textsc{Com} rule to be sound,
we maintain the invariant that if $\framedprovtrans{\Psi}{P}{\alpha}{\pi}{P'}$,
the outer binders of $\pi$ are precisely the binders of $\frameof{P}$.
Otherwise, the rules are exactly the same as in the original psi-calculi.

The reader may notice a curious asymmetry between the treatment of provenance binders
in the \textsc{ParL} and \textsc{ParR} rules.
This is to ensure that the order of the provenance binders coincides with the order of the frame
binders, and in the frame $\frameof{P \parop Q}$, the binders of $P$ occur syntactically
outside the binders of $Q$ (cf.~Definition~\ref{def:frame}).

\begin{exa}\label{exa:ether}
To illustrate how subjects and provenances interact,
we consider a psi-calculus where terms are names, assertions are
(finite) sets of names, composition is union, and channel connectivity
is defined as follows:

\[
\Psi \vdash x \chcon y \quad \mbox{iff} \quad x,y \in \Psi
\]

The intuition here is that there exists a single, shared communication
medium through which all communication happens. Processes are
allowed to declare aliases for this shared medium by
adding them to the assertion environment.

Consider the following processes, where $\pass{x}$ abbreviates $\pass{\{x\}}$
and $x \neq y$:

\begin{mathpar}
 P = (\nu x)(\overline{x} \parop \pass{x}) \and
 Q = (\nu y)(\underline{y} \parop \pass{y})
\end{mathpar}

\noindent Here $P$ and $Q$ are sending and receiving, respectively, via locally scoped
aliases of the shared communication medium. This example has been used
previously~\cite{bengtson.johansson.ea:psi-calculi-long},
to illustrate why the original psi-calculi needs channel equivalence
in all three of the \textsc{In}, \textsc{Out} and \textsc{Com} rules.
Up to scope extension $P \parop Q$ is equivalent to
\[
(\nu x,y)(\overline{x} \parop \pass{x} \parop \underline{y} \parop \pass{y})
\]
in which a communication between $x$ and $y$ is clearly possible,
because $x$ and $y$ are connected in the environment $\{x,y\}$.
Hence a communication must also be possible in $P \parop Q$. But
the two processes share no free names that can be used as communication subjects;
$P$ cannot do an output action with subject $x$ because $x$ is bound, and
similarly, $Q$ cannot do an input with subject $y$.
The only available option is for each of $P$ and $Q$ to derive transitions
labelled with the other's prefix:

\begin{mathpar}
  \and 
  \inferrule*[Left=\textsc{Scope}]
   {\inferrule*[Left=\textsc{Par-L}]
     {\inferrule*[Left=\textsc{Out}]
        {\{x,y\} \vdash x \chcon y}
        {\framedprovtrans{\{x,y\}}{\overline{x}}{\overline{y}}{x}{0}}
     }
     {\framedprovtrans{\{y\}}{\overline{x} \parop \pass{x}}{\overline{y}}{x}{0 \parop \pass{x}}}
   }
   {\framedprovtrans{\{y\}}{P}{\overline{y}}{(\nu x)x}{(\nu x)(0 \parop \pass{x})}}

  \and

  \inferrule*[Left=\textsc{Scope}]
   {\inferrule*[Left=\textsc{Par-L}]
     {\inferrule*[Left=\textsc{In}]
        {\{x,y\} \vdash x \chcon y}
        {\framedprovtrans{\{x,y\}}{\underline{y}}{\underline{x}}{y}{0}}
     }
     {\framedprovtrans{\{x\}}{\underline{y} \parop \pass{y}}{\underline{x}}{y}{0 \parop \pass{y}}}
   }
   {\framedprovtrans{\{x\}}{Q}{\underline{x}}{(\nu y)y}{(\nu y)(0 \parop \pass{y})}}
\end{mathpar}

In the original psi-calculi---where the exact same input and output transitions can be derived, but without the provenance annotations---it is clear that without the extra channel equivalence check in the \textsc{Com-Old} rule, we could not derive a communication between $P$ and $Q$.

With our provenance semantics the \textsc{Com} rule applies immediately.
Note that we have $\frameof{P} = (\nu x)\{x\}$ and $\frameof{Q} = (\nu y)\{y\}$,
and that $P$'s transition matches $Q$'s provenance and vice versa:

\begin{mathpar}
  \inferrule*[Left=\textsc{Com}]
    {\framedprovtrans{\{y\}}{P}{\overline{y}}{(\nu x)x}{(\nu x)(0 \parop \pass{x})} \\
     \framedprovtrans{\{x\}}{Q}{\underline{x}}{(\nu y)y}{(\nu y)(0 \parop \pass{y})}
    }
    {\framedprovtrans{\{\}}{P \parop Q}{\tau}{\bot}{(\nu x)(0 \parop \pass{x}) \parop (\nu y)(0 \parop \pass{y})}}
\end{mathpar}

\end{exa}

\begin{exa}
  This example is intended to illustrate how and why we maintain the invariant that frame
  and provenance binders coincide, and why it matters that they coincide in the
  \textsc{Com} rule.
  To this end, we will consider the process $P \parop Q$,
  where $P$ and $Q$ are defined as follows
  \begin{mathpar}
    P = (\nu x)((\nu y)\pass{\frass{P}} \parop !(\nu z)\outprefix{x}{z})
  \and
    Q = (\nu a b)(\inprefix{a}{c}{c}.R \parop \inprefix{b}{c}{c}.S \parop \pass{\frass{Q}})
  \end{mathpar}
  and where the composition $\frass{P} \otimes \frass{Q}$ entails the connectivity judgements
  $x \chcon a$ and $y \chcon b$, but not $x \chcon b$ or $y \chcon a$.
  We also assume $x,y \freshin \frass{Q}$ and $a,b \freshin \frass{P}$.
  Concretely, this can be realised by e.g.~extending the setup from Example~\ref{exa:ether}
  to use a pair of sets instead of a single set, and letting connectivity be membership
  in the same set.
  
  Let us focus on how we can derive a communication between the subjects $x$ and $a$.
  
  We have that $\frameof{P} = (\nu xy)\frass{P}$ and $\frameof{Q} = (\nu ab)\frass{Q}$.
  In the environment $\frass{P}$, we have the following derivation of an input transition from $Q$:
  
  \begin{mathpar}
  \inferrule*[Left=\textsc{Scope} $\times 2$]{
    \inferrule*[Left=\textsc{Par-L}]
      {\inferrule*[Left=\textsc{In}]
        {\frass{P} \otimes \frass{Q} \vdash x \chcon a}
        {
          \framedprovtrans%
          {\frass{P} \otimes \frass{Q}}
          {\inprefix{a}{c}{c}.R}
          {\inlabel{x}{z}}
          {a}
          {R\subst{c}{z}}          
        }
      }
      {\framedprovtrans%
        {\frass{P}}%
        {\inprefix{a}{c}{c}.R \parop \inprefix{b}{c}{c}.S \parop \pass{\frass{Q}}}%
        {\inlabel{x}{z}}%
        {a}%
        {R\subst{c}{z} \parop \inprefix{b}{c}{c}.S \parop \pass{\frass{Q}}}%
      }
    }
    {\framedprovtrans%
        {\frass{P}}%
        {Q}%
        {\inlabel{x}{z}}%
        {(\nu a,b;\epsilon)a}%
        {(\nu ab)(R\subst{c}{z} \parop \inprefix{b}{c}{c}.S \parop \pass{\frass{Q}})}%
    }    
  \end{mathpar}

  The corresponding output transition in $Q$ is derived as follows:

  \begin{mathpar}
    \inferrule*[Left={\textsc{Scope}}]
      {\inferrule*[Left=\textsc{Par-R}]
         {
          \inferrule*[Left=\textsc{Rep}]
           {
            \inferrule*[Left=\textsc{Par-L}]
             {             
              \inferrule*[Left=\textsc{Open}]
               {
                \inferrule*[Left=\textsc{Out}]
                 {
                  \frass{P} \otimes \frass{Q} \vdash x \chcon a
                 }
                 {\framedprovtrans
                   {\frass{P} \otimes \frass{Q}}
                   {\outprefix{x}{z}}
                   {\outlabel{a}{z}}
                   {x}
                   {0}
                 }
               }
               {\framedprovtrans
                  {\frass{P} \otimes \frass{Q}}
                  {(\nu z)\outprefix{x}{z}}
                  {\outlabel{a}{(\nu z)z}}
                  {(\nu z;\epsilon)x}
                  {0}
               }
             }
             {\framedprovtrans
               {\frass{P} \otimes \frass{Q}}
               {(\nu z)\outprefix{x}{z} \parop !(\nu z)\outprefix{x}{z}}
               {\outlabel{a}{(\nu z)z}}
               {(\nu z;\epsilon)x}
               {0 \parop !(\nu z)\outprefix{x}{z}}
             }
           }
           {\framedprovtrans
              {\frass{P} \otimes \frass{Q}}
              {!(\nu z)\outprefix{x}{z}}
              {\outlabel{a}{(\nu z)z}}
              {(\nu \epsilon;z)x}
              {0 \parop !(\nu z)\outprefix{x}{z}}
           }
         }
         {\framedprovtrans
           {\frass{Q}}
           {(\nu y)\pass{\frass{P}} \parop !(\nu z)\outprefix{x}{z}}
           {\outlabel{a}{(\nu z)z}}
           {(\nu y;z)x}
           {(\nu y)\pass{\frass{P}} \parop 0 \parop !(\nu z)\outprefix{x}{z}}
         }
      }
      {\framedprovtrans
        {\frass{Q}}
        {P}
        {\outlabel{a}{(\nu z)z}
        }
        {(\nu x,y;z)x}
        {(\nu x)((\nu y)\pass{\frass{P}} \parop 0 \parop !(\nu z)\outprefix{x}{z})}
      }
   \end{mathpar}

In the derivation above, notice how the provenance evolves throughout the derivation to
maintain the correspondence between the outer provenance binders and the frame binders.
Two rule applications are particularly noteworthy.
First, the \textsc{Rep} rule pushes $z$ from the outer binders to the inner binders,
because binders underneath the replication operator are not considered part of the
frame (cf.~Definition~\ref{def:frame}).
Second, the \textsc{Par-R} rule adds $y$, the frame binder of the leftmost parallel
component, to the outer binders.

Because both derivations have matching provenances and subjects,
and the frame binders and provenance binders used are the same,
the \textsc{Com} rule allows a derivation as follows:

\begin{mathpar}
\inferrule*[left=\textsc{Com}]
 {
 \framedprovtrans
        {\frass{Q}}
        {P}
        {\outlabel{a}{(\nu z)z}
        }
        {(\nu x,y;z)x}
        {\dots}
        \\
      {\framedprovtrans%
        {\frass{P}}%
        {Q}%
        {\inlabel{x}{z}}%
        {(\nu a,b;\epsilon)a}%
        {\dots}%
      }
  }
       {\framedprovtrans
         {\unit}
         {P \pll Q}
         {\tau}
         {\bot}
         {(\nu z)
           (\dots
            \parop
            \dots
           )}
       }
\end{mathpar}

In order for this rule to be sound, it is important that the frame binders
of $\frameof{P}$, and the provenance binders of the transition from $P$,
have the same ordering.
To see this, suppose we have a version of the \textsc{Com} rule which
allows transitions to be derived when frame and provenance binders are
equal up to reordering.
Call this alternative rule \textsc{Com'}.
We will now argue that \textsc{Com'} is unsound, because we
lose the ability to distinguish synchronisations
between $x$ and $a$ from synchronisations between $x$ and $b$.

First, note that since we identify frames up to alpha, we have
$\frameof{P} = (\nu yx)((x\;y)\cdot \frass{P})$.
By equivariance of $\vdash$ and $\chcon$  we have
\begin{mathpar}
  ((x\;y)\cdot \frass{P}) \otimes \frass{Q} \vdash x \chcon b
  \and
  ((x\;y)\cdot \frass{P}) \otimes \frass{Q} \vdash y \chcon a
\end{mathpar}
In this permuted frame, we can derive an input where $b$ receives from $x$:
\begin{mathpar}
  \framedprovtrans%
        {(x\;y)\cdot \frass{P}}%
        {Q}%
        {\inlabel{x}{z}}%
        {(\nu a,b;\epsilon)b}%
        {(\nu ab)(\inprefix{a}{c}{c}.R \parop S\subst{c}{z} \parop \pass{\frass{Q}})}
\end{mathpar}
Since we identify provenances up to alpha, we also have
\begin{mathpar}
  \framedprovtrans%
        {(x\;y)\cdot \frass{P}}%
        {Q}%
        {\inlabel{x}{z}}%
        {(\nu b,a;\epsilon)a}%
        {(\nu ab)(\inprefix{a}{c}{c}.R \parop S\subst{c}{z} \parop \pass{\frass{Q}})}
\end{mathpar}

Using this transition, and the same derivation of a transition from $P$ as above, we can now apply \textsc{Com'} to derive a synchronisation between $x$ and $b$, despite the fact that $x$ and $b$ are not connected:

\begin{mathpar}
\inferrule*[left=\textsc{Com'}]
 {
 \framedprovtrans
        {\frass{Q}}
        {P}
        {\outlabel{a}{(\nu z)z}
        }
        {(\nu x,y;z)x}
        {\dots}
        \\
      {\framedprovtrans%
        {(x\;y)\cdot \frass{P}}%
        {Q}%
        {\inlabel{x}{z}}%
        {(\nu b,a;\epsilon)a}%
        {\dots}
      }
  }
       {\framedprovtrans
         {\unit}
         {P \pll Q}
         {\tau}
         {\bot}
         {(\nu z)
           (\dots
            \parop
            \dots
           )}
       }
\end{mathpar}

If we push all binders in $P \parop Q$ to the top level,
no such derivation is possible. Thus scope extension fails to hold
with \textsc{Com'}.

With \textsc{Com}, the existence of this alternative transition from $Q$
is unproblematic: it cannot synchronise with any transition from $P$ unless
that transition too uses the permuted frame.

With the same counterexample,
we can also see why it is important that the provenance retains all of the
frame binders, even the vacuous ones: the provenances $(\nu x)x$ and $(\nu y)y$
are equal, so the provenance would contain no information about which prefix
the transition originates from.
\end{exa}



\section{Meta-theory}\label{sec:metatheory}

In this section, we will derive the standard algebraic and congruence laws of strong
and weak bisimulation,
develop an alternative formulation of strong bisimulation in terms of a reduction relation and
barbed congruence, and show that our extension of psi-calculi is conservative.

\subsection{Strong bisimulation}\label{sec:bisimulation}

We write $\framedtrans{\Psi}{P}{\alpha}{P'}$ as shorthand for $\exists \pi.\;\framedprovtrans{\Psi}{P}{\alpha}{\pi}{P'}$. Bisimulation is then defined exactly as in the original psi-calculi:

\begin{defi}[Strong bisimulation]\label{def:strongbisim} A symmetric relation $\mathcal{R} \subseteq \assertions \times \processes \times \processes$ is a \emph{strong bisimulation} iff for every $(\Psi,P,Q) \in \mathcal{R}$
\begin{enumerate}
  \item $\Psi \otimes \frameof{P} \;\simeq\; \Psi \otimes \frameof{Q}$ (static equivalence)
  \item $\forall \Psi'. (\Psi\otimes\Psi',P,Q) \in \mathcal{R}$ (extension of arbitrary assertion)
  \item If $\framedtrans{\Psi}{P}{\alpha}{P'}$ and $\bn{\alpha} \freshin \Psi, Q$, then there exists $Q'$ such that $\framedtrans{\Psi}{Q}{\alpha}{Q'}$ and $(\Psi,P',Q') \in \mathcal{R}$ (simulation)
\end{enumerate}

We let \emph{bisimilarity} $\bisim$ be the largest bisimulation. We write $\trisimsub{\Psi}{P}{Q}$ to mean $(\Psi,P,Q) \in\;\bisim$, and $P \bisim Q$ for $\trisimsub{\one}{P}{Q}$.
\end{defi}

Clause 3 is the same as for pi-calculus bisimulation.
Clause 1 requires that two bisimilar processes expose statically equivalent
assertion environments.
Clause 2 states that if two processes are bisimilar in an environment, they must be bisimilar
in every extension of that environment. Without this clause, bisimulation is not
preserved by parallel composition.

This definition might raise some red flags for the experienced concurrency theorist.
We allow the matching transition from $Q$ to have any provenance,
irrespective of what $P$'s provenance is.
Hence the \textsc{Com} rule uses information that is ignored for the purposes of bisimulation,
which in most cases would result in a bisimilarity that is not preserved by the parallel operator.

Before showing that bisimilarity is nonetheless compositional,
we will argue that bisimilarity would be too strong if Clause 4 required transitions with
matching provenances.
Consider two distinct terms $M,N$ that are connected to the same channels;
that is, for all $\Psi,K$ we have $\Psi \vdash M \chcon K$ iff $\Psi \vdash N \chcon K$.
We would expect $\overline{M}.0$ and $\outprefix{N}.0$ to be bisimilar because
they offer the same interaction possibilities.
With our definition, they are.
But if bisimulation cared about provenance they would be distinguished, because
transitions originating from $\overline{M}.0$ will have provenance $M$
while those from $\outprefix{N}.0$ will have $N$.

The key intuition is that what matters is not which provenance a transition has, but
which channels the provenance is connected to.
The latter is preserved by Clause 3,
as this key technical lemma hints at:

\begin{lem}(Find connected provenance)\label{lemma:provchaneq}
  \begin{enumerate}
  \item If $\framedprovtrans{\Psi}{P}{\inlabel{M}{N}}{\pi}{P'}$ and $C$ is a finitely supported nominal set, then there exists $\frnames{P},\frass{P},\ve{x},K$ such that $\frameof{P} = \framepair{\frnames{P}}{\frass{P}}$ and $\pi = (\nu\frnames{P};\ve{x})K$ and $\frnames{P} \freshin \Psi,P,M,N,P',C,\ve{x}$ and $\ve{x} \freshin \Psi,P,N,P',C$ and $\Psi \otimes \frass{P} \vdash M \chcon K$.
  \item A similar property for output transitions (elided).
  \end{enumerate}
\end{lem}

\begin{proof}
  Formally proven in Isabelle, by a routine induction.
\end{proof}

\noindent In words,
the provenance of a transition is always connected to its subject,
and the frame binders can always be chosen sufficiently fresh for any context.
This simplifies the proof that bisimilarity is preserved by parallel:
in the original psi-calculi, one of the more challenging aspects of this proof is
finding sufficiently fresh subjects to use in the \textsc{Com-Old} rule,
and then using associativity and symmetry to connect them
(cf.~\cite[Lemma 5.11]{bengtson.johansson.ea:psi-calculi-long}).
By Lemma~\ref{lemma:provchaneq} we already have a sufficiently
fresh subject: our communication partner's provenance.


\begin{thm}[Congruence properties of strong bisimulation]\label{thm:bisimpres}\
\begin{enumerate}
  \item $\trisimsub \Psi P Q \quad \Rightarrow \quad \trisimsub{\Psi}{P \parop R}{Q \parop R}$
  \item $\trisimsub \Psi P Q \quad \Rightarrow \quad \trisimsub{\Psi}{(\nu x)P}{(\nu x)Q}$ if $x \freshin \Psi$
  \item $\trisimsub \Psi {P_G}{Q_G} \quad \Rightarrow \quad \trisimsub{\Psi}{! P_G}{\;! Q_G}$
  \item \label{case:case} $\forall i. \trisimsub {\Psi}{P_i}{Q_i} \quad \Rightarrow \quad \trisimsub{\Psi}{\caseonly\;{\ci{\vec{\varphi}}{\vec{P}}}}{\caseonly\;{\ci{\vec{\varphi}}{\vec{Q}}}}$ if $\vec{P}, \vec{Q}$ are assertion-guarded
  \item $\trisimsub \Psi P Q \quad \Rightarrow \quad \trisimsub{\Psi}{\outprefix{M}{N}.P}{\outprefix{M}{N}.Q}$
\end{enumerate}
\end{thm}

\begin{proof}
  Formally proven in Isabelle. All proofs are by coinduction.
  The most interesting cases are parallel and replication,
  where Lemma~\ref{lemma:provchaneq} features prominently.
  We briefly outline a \textsc{Com} subcase of the replication case,
  where the candidate relation is
  \[\{(\Psi, R \parop !P, R \parop !Q). \trisimsub \Psi P Q \wedge \mbox{$P,Q$ are assertion guarded}\}\]
  Suppose $P \bisim Q$ and that $!P$ derives a $\tau$ transition from
communication between two unfolded copies of $P$, with
input subject $M$ and output subject $K$.
We need to mimic the same communication between two copies of $Q$, but after
using $\trisimsub \Psi P Q$ to obtain a matching input transition, the subject $M$ is not useful
to derive a communication since it is $P$'s provenance, not $Q$'s.
However, we can obtain eligible subjects $M',K'$ by repeatedly applying Lemma~\ref{lemma:provchaneq}.
\end{proof}

In Theorem~\ref{thm:bisimpres}.\ref{case:case}, $P_i$ is the $i$:th element of
$\vec{P}$, and similarly for $Q_i$. The index variable $i$ ranges over the
length of the sequences $\vec{\varphi},\vec{P},\vec{Q}$, which we assume are
equal.

\begin{thm}[Algebraic laws of strong bisimulation]\label{thm:strong-struct}\
  
\begin{enumerate}
  \item $\trisimsub{\Psi}{P}{P \parop \nil}$
 \item $\trisimsub{\Psi}{P\parop ( Q \parop R)}{(P \parop Q) \parop R}$
 \item $\trisimsub{\Psi}{P \parop Q}{Q \parop P}$
 \item $\trisimsub{\Psi}{(\nu a)\nil}{\nil}$
 \item $\trisimsub{\Psi}{P \parop (\nu a) Q}{(\nu a)(P \parop Q)}\mbox{  if $a \freshin P$}$
 \item $\trisimsub{\Psi}{\outprefix{M}{N}.(\nu a)P}{(\nu a)\outprefix{M}{N}.P}\mbox{  if $a \freshin M, N$}$
 \item $\trisimsub{\Psi}{\inprefix{M}{\vec{x}}{N}.(\nu a)P}{(\nu a)\inprefix{M}{\vec{x}}{N}.P}\mbox{  if $a \freshin \vec{x},M,N$}$
 \item $\trisimsub{\Psi}{!P}{P \parop !P}$
 \item $\trisimsub{\Psi}{\caseonly\;{\ci{\vec{\varphi}}{\vec{(\nu a)P}}}}{(\nu a)\caseonly\;{\ci{\vec{\varphi}}{\vec{P}}}}\mbox{  if $a \freshin \vec{\varphi}$}$
 \item $\trisimsub{\Psi}{(\nu a)(\nu b)P}{(\nu b)(\nu a)P}$
\end{enumerate}
\end{thm}

\begin{proof}
  Formally proven in Isabelle. All proofs are by coinduction.
\end{proof}

\noindent
Note that bisimilarity is not preserved by input, for the same reasons as the \pic.
As in the \pic, we can define \emph{bisimulation congruence} as the substitution
closure of bisimilarity, and thus obtain a true congruence which satisfies all the
algebraic laws above.
We have verified this in Nominal Isabelle, following~\cite{bengtson.johansson.ea:psi-calculi-long}.

\subsection{Weak bisimulation}\label{sec:weakbisimulation}

We have also proved the standard algebraic and congruence properties of weak bisimulation.
The results in this section were established for the original psi-calculi by
Johansson et al.~\cite{DBLP:conf/lics/JohanssonBPV10}; our contribution is to lift them to
psi-calculi without channel symmetry and transitivity. As for strong bisimulation,
it turns out that we may disregard provenances for the purposes of weak bisimulation, so we
can reuse the original definitions verbatim.

The definition of weak bisimulation is technically complicated in psi-calculi because of the
delicate interplay between assertions and reductions. For example, in the
pi-calculus weak bisimulation equates $P$ and $\tau.P$, but in psi-calculi this equation
cannot be admitted: $P$ may contain top-level assertions that disable interaction possibilities in parallel processes. Hence there may be situations where an observable action originating
from $Q$ is available in $Q \parop \tau.P$ (where $P$ has not yet disabled it)
but unavailable in $Q \parop P$.

For a comprehensive motivation of the definitions, we refer to Johansson et al.~\cite{DBLP:conf/lics/JohanssonBPV10}; below we restate the pertinent definitions for completeness.

\begin{defi}[Weak transitions]
\label{def:weakTrans}
$\Psi \frames \wtrans{P}{}{P'}$ means that either $P = P'$ or there exists $P''$ such that $\Psi \frames \trans{P}{\tau}{P''}$ and $\Psi \frames \wtrans{P''}{}{P'}$.
We write $\Psi \frames \wtrans{P}{\alpha}{P'}$ to mean that there exists $P'',P'''$ such that
$\Psi \frames \wtrans{P}{}{P''}$ and $\framedtrans{\Psi}{P''}{\alpha}{P'''}$ and
$\Psi \frames \wtrans{P'''}{}{P'}$.
\end{defi}
\begin{defi}
$P$ statically implies $Q$ in the environment $\Psi$, written $P \simplies_\Psi Q$, if\[%
\forall \varphi. \; \Psi \ftimes \frameof{P} \vdash \varphi \;\Rightarrow \;
\Psi \ftimes \frameof{Q} \vdash \varphi\]
If $\Psi = \unit$ we may write $P\simplies Q$.
\end{defi}

\begin{defi}[Weak bisimulation]
A {\em weak bisimulation}
 $\mathcal R$ is a ternary relation between assertions and pairs of agents such that
 ${\mathcal R}(\Psi,P,Q)$ implies all of
 \begin{enumerate}
 \item Weak static implication: 
\[\begin{array}{l}
\forall \Psi' \exists Q'', Q'. \\
 \quad \Psi \frames \wtrans{Q}{}{Q''}\quad
\wedge\quad P \simplies_\Psi Q'' \quad  \wedge\\
\quad \Psi \ftimes \Psi' \frames \wtrans{Q''}{}{Q'}  \quad \wedge \quad
 {\mathcal R}(\Psi\ftimes\Psi',P,Q')
\end{array}
\]
 \item
   Symmetry: ${\mathcal R}(\Psi,Q,P)$
 \item
 Extension of arbitrary assertion:\\  
$\forall \Psi'.\ {\mathcal R}(\Psi \ftimes \Psi',P,Q)$
 \item   Weak simulation:
for all $\alpha, P'$ such that $\bn{\alpha}\freshin \Psi,Q$ and $\framedtrans{\Psi}{P}{\alpha}{P'}$ it holds
\[\begin{array}{@{}lrl@{}}
\mbox{if}\;\alpha = \tau: \exists Q' . 
\;\Psi \frames \wtrans{Q}{}{Q'} \quad \wedge \quad {\mathcal R}(\Psi,P',Q') \\
\mbox{if}\;\alpha \neq \tau: \forall \Psi' \exists Q'', Q'''. & &\\
\quad\Psi \frames  \wtrans{Q}{}{Q'''}\quad \wedge \quad  P \simplies_\Psi Q'''
\quad \wedge\\
\quad \Psi \frames \trans{Q'''}{\alpha}{Q''} \quad \wedge \\
\quad  \exists Q'. \;\Psi \ftimes \Psi' \frames \wtrans{Q''}{}{Q'} 
 \;\; \wedge \;\; {\mathcal R}(\Psi\ftimes\Psi',P',Q')
\end{array}\]
\end{enumerate}
\label{def:noweakwbisim}
We define $P \wbisim_\Psi Q$ to mean that there exists a weak bisimulation ${\mathcal R}$
such that ${\mathcal R}(\Psi,P,Q)$ and write $P \wbisim Q$ for $P \wbisim_\unit Q$.
\end{defi}

Weak bisimulation thus defined includes strong bisimulation, and thus
satisfies all the usual structural laws.
It is not preserved by $\caseonly$ and input, for the same reasons
as $+$ and input do not preserve weak bisimulation in the pi-calculus.
We employ the standard solution to obtain a congruence:
all initial $\tau$ steps must be
simulated by at least one $\tau$ step, and we furthermore close the relation under all
substitutions.

\begin{defi}[Weak congruence]
$P$ and $Q$ are weakly $\tau$-bisimilar, written $\taubisim{\Psi}{P}{Q}$, if $P
\wbisim_\Psi Q$ and the following holds:
for all $P'$ such that
$\framedtrans{\Psi}{P}{\tau}{P'}$ there exists $Q'$ such that
$\Psi\frames \wtrans{Q}{\tau}{Q'}  \; \wedge \; P'\wbisim_{\Psi} Q'$,
and similarly with the roles of $P$ and $Q$ exchanged. We define $P \wcong Q$ to
mean that for all $\Psi$, and for all well-formed substitution sequences $\ve{\sigma}$,
it holds that
$\taubisim{\Psi}{P\ve{\sigma}}{Q\ve{\sigma}}$.
     \end{defi}

The following theorems have been formally proven in Nominal Isabelle:

\begin{thm}\label{thm:weak-struct}$\underset{{\rm tau}}{\wbisim}$ satisfies all the algebraic laws of
  $\bisim$ established in Theorem~\ref{thm:strong-struct}.
\end{thm}

\begin{proof}
  Formally proven in Isabelle. The proof relies on the fact that
  ${\bisim} \subseteq {\wbisim}$, which we show by coinduction,
  using $\bisim$ as a candidate relation.
\end{proof}

\begin{thm}\label{thm:weak-bisim-cong}$\wbisim$ satisfies all the congruence properties of  $\bisim$ established in Theorem~\ref{thm:bisimpres} except~\ref{thm:bisimpres}.4.
\end{thm}

\begin{proof}
  Formally proven in Isabelle, by coinduction.
\end{proof}

\begin{thm}\label{thm:weak-cong-cong}Weak congruence $\wcong$ is a congruence wrt.~all operators of psi-calculi.
\end{thm}

\begin{proof}
  Formally proven in Isabelle, using Theorem~\ref{thm:weak-bisim-cong} where applicable.
\end{proof}

\subsection{Motivating the design}\label{sec:design}

We have added provenance annotations to an operational semantics that had no shortage
of annotations and side-conditions to begin with. The end result may strike the reader
as somewhat unparsimonious. Previously, psi-calculi had one label component---the channel
subjects---for keeping track of connectivity. We now have two; do we really need both?
In this section, we will explore the consequences of removing either channel subjects or
provenances from the semantics.
The short answer is that while we \emph{can} do without either one, the end result is
not greater parsimony.

\subsubsection{Do we need provenances?}

The fact that bisimilarity is compositional yet ignores provenances
suggests that the semantics could be reformulated without provenance annotations on labels.
To achieve this, what is needed is a side-condition $S$ for the \textsc{Com} rule which,
given an input and an output with subjects $M,K$,
determines if the input transition could have been derived from prefix $K$, and vice versa:

\begin{mathpar}
\inferrule*[]
 {\framedtrans{\frass{Q} \ftimes \Psi}{P}{\bout{\ve{a}}{M}{N}}{P'} \\
  \framedtrans{\frass{P} \ftimes \Psi}{Q}{\inlabel{K}{N}}{Q'} \\
  S
  }
       {\framedtrans{\Psi}{P \pll Q}{\tau}{(\nu \ve{a})(P' \pll Q')}}
\end{mathpar}

\noindent But we already have such an $S$: the semantics \emph{with} provenances! So we can let
\[S = \framedprovtrans{\frass{Q} \ftimes \Psi}{P}{\outlabel{M}{(\nu \ve{a})N}}{(\nu\frnames{P};\ve{x})K}{P'} \wedge \framedprovtrans{\frass{P} \ftimes \Psi}{Q}{\inlabel{K}{N}}{(\nu\frnames{Q};\ve{y})M}{Q'}\]

\noindent Of course, this definition is not satisfactory: the provenances are still there, just swept under the
carpet. Worse, we significantly complicate the definitions by effectively introducing a stratified
semantics. Thus the interesting question is not whether such an $S$ exists (it does), but whether $S$
can be formulated in a way that is significantly simpler than the semantics with provenances.
The author believes the answer is negative: $S$ is a property about the roots of the
proof trees used to derive the transitions from $P$ and $Q$. The provenance records just enough
information about the proof trees to show that $M$ and $K$ are connected; with no provenances,
it is not clear how this information could be obtained without essentially reconstructing
the proof tree.

Another alternative is to use the proof tree itself as the transition label~\cite{boudol1988non,DBLP:conf/icalp/DeganoP92}. This makes the necessary information available, at the expense of making labels even more complicated.

\subsubsection{Do we need channel subjects?}

While we have chosen to test for channel connectivity in the \textsc{In} and
\textsc{Out} rules,
a semantics without channel subjects would defer the connectivity check until
the \textsc{Com} rule.%
\footnote{This is similar to a design first proposed by Magnus Johansson in an unpublished draft.
Johansson's design does not use provenances, but obtains a similar effect by including bound
subjects and bound assertions in labels. By partitioning provenance binders in two sequences,
we can recover frame binders from labels and thus found no need to include bound assertions.}
Let us call the former approach \emph{early connectivity}, and the latter \emph{late connectivity}.
The rules for late connectivity---eliding freshness conditions for readability, and
using $!$ and $?$ to distinguish outputs from inputs---would be:

\begin{mathpar}

\inferrule*[Left=\textsc{In-Late}]
    {\,}   
{\framedprovtrans{\Psi}{\inprefix{M}{\ve{y}}{N} \sdot
P}{?{N}\subst{\ve{y}}{\ve{L}}}{M}{P\subst{\ve{y}}{\ve{L}}}}

\inferrule*[left=\textsc{Out-Late}]
    {\,}
    {\framedprovtrans{\Psi}{\outprefix{M}{N} \sdot P}{!{N}}{M}{P}}

\inferrule*[left=\textsc{Com-Late}, right={$\ve{a} \freshin Q$}]
 {\Psi \ftimes \frass{P} \ftimes \frass{Q} \vdash K \chcon M \\
 \framedprovtrans{\frass{Q} \ftimes \Psi}{P}{!{(\nu \ve{a})N}}{(\nu\frnames{P};\ve{x})K}{P'} \\
  \framedprovtrans{\frass{P} \ftimes \Psi}{Q}{?{N}}{(\nu\frnames{Q};\ve{y})M}{Q'} 
  }
       {\framedprovtrans{\Psi}{P \pll Q}{\tau}{\bot}{(\nu \ve{a})(P' \pll Q')}}
    
\end{mathpar}

It is pleasant that this formulation allows the \textsc{In-Late} and \textsc{Out-Late} rules
to have no side-conditions. Save for the provenance bookkeeping, all questions
of connectivity are handled entirely in \textsc{Com-Late}, which results in a pleasing
separation of concerns. This may seem more parsimonious at first glance, but it introduces two issues
that makes the trade-off seem unfavourable: more complicated bisimulation and spurious transitions.

\begin{enumerate}

\item More complicated bisimulation. Consider a psi-calculus where channel connectivity is syntactic equality; that is, where
$\Psi \vdash M \chcon K$ holds iff $M=K$.
Fix $M,K$ such that $M\neq K$. Using bisimilarity as defined in Definition~\ref{def:strongbisim},
we can show that $\overline{M}.0 \bisim \overline{K}.0$: without subjects, these processes emit identical labels save for the provenance, which is ignored by bisimulation.
Hence bisimilarity fails to be preserved by the parallel operator: consider these processes in parallel with a process that can receive on $M$. Then $\overline{M}.0$ can communicate but $\overline{K}.0$ cannot.

The takeaway is that with late connectivity, a compositional notion of bisimulation needs to be
more careful with which provenance the mimicking transition may use. The intuition is that rather
than admitting any provenance, we admit provenances that are connected to the same channels.
We conjecture that the necessary adaptation is:

\begin{defi}[Late-connectivity bisimulation]\label{def:cfstrongbisim} A symmetric relation $\mathcal{R} \subseteq \assertions \times \processes \times \processes$ is a \emph{channel-free bisimulation} iff for every $(\Psi,P,Q) \in \mathcal{R}$
\begin{enumerate}
  \item $\Psi \otimes \frameof{P} \;\simeq\; \Psi \otimes \frameof{Q}$ (static equivalence)
  \item $\forall \Psi'. (\Psi\otimes\Psi',P,Q) \in \mathcal{R}$ (extension of arbitrary assertion)
  \item If $\framedprovtrans{\Psi}{P}{\alpha}{\pi}{P'}$ and $\bn{\alpha} \freshin \Psi, Q$, then
    \begin{enumerate}
      \item If $\alpha=\tau$ then there exists $Q'$ such that $\framedprovtrans{\Psi}{Q}{\tau}{\bot}{Q'}$ and $(\Psi,P',Q') \in \mathcal{R}$
      \item For all $M,K,N,\ve{x},\frnames{P},\frass{P},\frnames{Q},\frass{Q}$ such that $\alpha = ?N$ and $\pi = (\nu\frnames{P};\ve{x})K$ and $\frameof{P} = \framepair{\frnames{P}}{\frass{P}}$ and $\frameof{Q} = \framepair{\frnames{Q}}{\frass{Q}}$ and $\frnames{P},\frnames{Q} \freshin \Psi,P,Q,M,\ve{x}$ and $\ve{x}\freshin\Psi$ and $\Psi \otimes \Psi_P \vdash M \chcon K$, then there exists $\ve{y},K',Q'$ such that $\ve{y} \freshin \Psi$ and $\Psi \otimes \Psi_Q \vdash M \chcon K'$ and $\framedprovtrans{\Psi}{Q}{?N}{(\nu\frnames{Q};\ve{y})K'}{Q'}$ and $(\Psi,P',Q') \in \mathcal{R}$
      \item (A similar clause for output transitions)
    \end{enumerate}
\end{enumerate}
\end{defi}

In words, for every channel $M$ that the transition from $P$'s provenance is connected to, there
must be a transition from $Q$ with a provenance that is also connected to $M$. Note that static
equivalence is not sufficient to imply preservation of connectivity: the conditions
may be distinct, and even if equal may be obscured by the frame binders.

We find this definition of bisimulation intolerably ad-hoc and complicated.

\item Spurious transitions. Consider the representation of the $\pi$-calculus in psi-calculi, where $\terms = \nameset$, and where channel connectivity is syntactic equality on names. This representation is in one-to-one transition correspondence with the standard presentation of the $\pi$-calculus~\cite{Johansson10}, but if we use late connectivity, one-to-one transition correspondence is lost. The pi-calculus process $(\nu x)\overline{x}y$ should not have any outgoing transitions, but late connectivity semantics allows the derivation of a transition as follows:

\begin{mathpar}
\inferrule*[left=\textsc{Scope-Late}]
    {
    \inferrule*[left=\textsc{Out-Late}]
      {\,}
      {\framedprovtrans{\Psi}{\overline{x}y}{!y}{x}{0}}      
    }
    {\framedprovtrans{\Psi}{(\nu x)\overline{x}y}{!y}{(\nu x)x}{(\nu x)0}}
\end{mathpar}

\noindent The existence of this transition is more of a blemish than a real problem.
It cannot be used to derive a communication because there exists no $y$ such that
$x \neq y$ and $x \chcon y$. It will be ignored by late-connectivity bisimulation
(Definition~\ref{def:cfstrongbisim}) for the same reason, so it remains true
that $(\nu x)\overline{x}y \bisim 0$. Still, we maintain the view that a derivable
transition should signify the readiness of the process to engage in some behaviour.
This transition signifies nothing.

\end{enumerate}


\subsection{Revisiting the counterexamples}\label{sec:counterexamples}

In the introduction, we mentioned that Bengtson et al.~\cite{bengtson.johansson.ea:psi-calculi-long} have counterexamples to the effect that without symmetry and transitivity,
scope extension is unsound. In this section we will revisit these counterexamples, with the
aim of convincing the reader that they do not apply to our developments in the present
paper.

We begin by quoting the counterexample used by Bengtson et al. to argue that scope extension requires channel symmetry~\cite[p.~14]{bengtson.johansson.ea:psi-calculi-long}:

\begin{quote}
  Consider any psi-calculus where $\Psi_1$ and $\Psi_2$ are such that $\Psi_1 \ftimes \Psi_2
\vdash a \sch b$ and $\Psi_1 \ftimes \Psi_2 \vdash b \sch b$. We shall argue that also $\Psi_1 \ftimes \Psi_2 \vdash b \sch a$ must hold, otherwise scope extension does not hold. 
Consider the agent
\[(\nu a,b)(\pass{\Psi_1}\parop\pass{\Psi_2}\parop\overline{a}\sdot \nil\parop
b\sdot \nil)\]
which has an internal communication $\tau$ using $b$ as subjects in the premises of the {\sc Com} rule. 
If $b \freshin \Psi_1$ and $a \freshin \Psi_2$, by scope extension the agent should behave as
\[(\nu a)(\pass{\Psi_1}\parop\overline{a}\sdot \nil) \;\parop\; (\nu b)(\pass{\Psi_2}\parop
b\sdot \nil)\]
and therefore this agent must also have a $\tau$ action. The left hand component cannot do an $\overline{a}$ action, but in the environment of $\Psi_2$ it can do a $\overline{b}$ action.
Similarly, the right hand component cannot do a $b$ action. The only possibility is for it to do an $a$ action, as in
\[\framedtrans{\Psi_1}{(\nu
b)(\pass{\Psi_2} \parop b \sdot \nil)}{a}{\cdots}\]
and this requires $\Psi_1 \ftimes \Psi_2 \vdash b \sch a$.
\end{quote}

\noindent This counterexample is only valid because the authors were not careful about the orientation of channel equivalence judgements in the operational rules---understandably so, because they were designing a calculus with symmetric connectivity. To see this clearly, consider the preconditions to the rules for input and output in the original psi-calculi:

\begin{mathpar}
  \and 
\inferrule*[Left=\textsc{In-Old}]
    {\Psi \vdash M \sch K }
    {\framedtrans{\Psi}{\inprefix{M}{\ve{y}}{N}.P}{\inlabel{K}{N}\lsubst{\ve{L}}{\ve{y}}}{P\lsubst{\ve{L}}{\ve{y}}}}

\inferrule*[left=\textsc{Out-Old}]
    {\Psi \vdash M \sch K }
    {\framedtrans{\Psi}{\outprefix{M}{N}.P}{\outlabel{K}{N}}{P}}
\end{mathpar}

\noindent Note the inconsistency that in \textsc{In-Old}, the input channel occurs on the LHS of $\sch$, whereas in \textsc{Out-Old} the output channel occurs on the LHS. This of course makes no difference when $\sch$ is symmetric, but for an asymmetric connectivity relation it is important to use it consistently, with the input channel always going on the same side of $\sch$. Simply reorienting the channel equivalence judgement in \textsc{In-Old} suffices to make this counterexample
inapplicable even to the original psi-calculi.
This should not be taken to mean that the original psi-calculi do not require symmetry: we only
mean to say that the reason for the symmetry requisite is not clear from this counterexample.

For transitivity, Bengtson et al. give the following counterexample~\cite[p.~14]{bengtson.johansson.ea:psi-calculi-long}:

\begin{quote}
Let $\one$ entail $a \sch a$ for all names $a$, and
let $\Psi$ be an assertion with support $\{a,b,c\}$ that additionally entails
the two conditions $a \sch b$ and $b \sch c$, but not $a \sch c$, and thus does
not satisfy transitivity of channel equivalence. If $\Psi$ entails no other
conditions then $(\nu b)\Psi \sequivalent \one$, and we expect $(\nu b)\pass\Psi$
to be interchangeable with $\pass\one$ in all contexts.
Consider the agent
\[\overline{a}\sdot \nil \parop c \sdot \nil \parop (\nu b)\pass{\Psi}\]
By scope extension it should behave precisely as
\[(\nu b)(\overline{a}\sdot \nil \parop c \sdot \nil \parop \pass\Psi)\]
This agent has a $\tau$-transition since  $\Psi$ enables an interaction between the components $\overline{a}\sdot \nil$ and $c \sdot \nil$.
But the agent
\[\overline{a}\sdot \nil \parop c \sdot \nil \parop \pass{\one}\]
has no such transition. The conclusion is that $(\nu b) \Psi$ must entail that the components can communicate, i.e.~that $a \sch c$, in other words $\Psi \vdash a \sch c$.
\end{quote}

\noindent The present author agrees with this reasoning, but reaches the exact opposite conclusion about
which process is at fault: the anomaly here is not that $\overline{a}\sdot \nil \parop c \sdot \nil \parop \pass{\one}$ cannot reduce, but that $(\nu b)(\overline{a}\sdot \nil \parop c \sdot \nil \parop \pass\Psi)$ can. If $a$ and $c$ are not channel equivalent, there should not be a derivable communication between the channels $a$ and $c$. The original psi-calculi nonetheless admit the derivation

\begin{mathpar}
\inferrule*[left=\textsc{Scope}]
 {\inferrule*[left=\textsc{Par}]
  {
   \inferrule*[left=\textsc{Com-Old}]
    {{\inferrule*[Left=\textsc{Out}]
      {\Psi \vdash a \sch b}
      {\framedtrans{\Psi}{\overline{a}\sdot \nil}{\overline{b}}{\cdots}}} \\
     {\inferrule*[Left=\textsc{In}]
      {\Psi \vdash c \sch c}
      {\framedtrans{\Psi}{c \sdot \nil}{\underline{c}}{\cdots}}} \\
     \Psi \vdash b \sch c
    }
    {\framedtrans{\Psi}{\overline{a}\sdot \nil \parop c \sdot \nil}{\tau}{\cdots}
    }
  }
  {\framedtrans{\one}{\overline{a}\sdot \nil \parop c \sdot \nil \parop \pass\Psi}{\tau}{\cdots}}
 }
 {\framedtrans{\one}{(\nu b)(\overline{a}\sdot \nil \parop c \sdot \nil \parop \pass\Psi)}{\tau}{\cdots}}
\end{mathpar}

Notice how we use three different channel equivalence judgements to string together
a derivation via $b$, which both $a$ and $c$ are connected to. This is not a problem if
transitivity is intended, but leads to absurd derivations if the intention is to allow
non-transitive connectivity relations.

With the provenance semantics, the counterexample does not apply since no communication between $a$ and $c$ is possible: the only possibility is to apply the \textsc{Com} rule with matching subjects and provenances. This requires $\overline{a}\sdot \nil$ to have an output transition with subject $c$ and $c\sdot \nil$ to have an input transition with subject $a$, but such transitions cannot be derived because $a \sch c$ (in our notation $a \chcon c$) does not hold.

Finally, we observe that a slight variant of the counterexample for transitivity illustrates the
need for symmetry.
This time, let $\Psi$ be an assertion with support $\{a,b,c,d\}$ that entails the three conditions $a \sch b$ and $d \sch b$ and $c \sch d$ and none other.
The agent
\[\overline{a}\sdot \nil \parop c \sdot \nil \parop \one\]
will have no transitions, but this agent will have a $\tau$ transition:
\[(\nu b,d)(\overline{a}\sdot \nil \parop c \sdot \nil \parop \pass\Psi)\]
We conclude by the same reasoning as above that $\Psi \vdash a \sch c$ must hold,
or in other words, that channel equivalence must satisfy the law
\[a \sch b \wedge d \sch b \wedge c \sch d \Rightarrow a \sch c\]
\noindent This awkward-looking algebraic law is weaker than symmetry and transitivity,
and together with reflexivity it implies both.
It may well be that this weaker law is sufficient for the original psi-calculi to be
compositional (assuming a channel equivalence reorientation in the \textsc{In-Old} rule).
However, we find it difficult to imagine a useful connectivity relation that satisfies this law
but is neither reflexive, transitive nor symmetric.

\subsection{Validation}\label{sec:validation}

We have defined semantics and bisimulation, and showed that bisimilarity satisfies the expected
laws. But how do we know that they are the right
semantics, and the right bisimilarity?
This section provides two answers to this question. First, we show that our developments
constitute a conservative extension of the original psi-calculi.
Second, we define a reduction semantics and barbed bisimulation that are
in agreement with our (labelled) semantics and (labelled) bisimilarity.

Let $\apitransarrow{}_o$ and $\bisim_o$ denote semantics and bisimilarity as defined
by Bengtson et al.~\cite{bengtson.johansson.ea:psi-calculi-long},
i.e., without provenances
and with the \textsc{Com-Old} rule discussed in Section~\ref{sec:definitions}.
Along the same lines, let $\wbisim_o$ and ${\underset{{\rm tau}}{\wbisim}}_o$ and $\wcong_o$
denote weak bisimulation, weak $\tau$-bisimilarity and weak congruence as defined by
Johansson et al.~\cite{DBLP:conf/lics/JohanssonBPV10}.
Then conservativity can be stated as follows:

\begin{thm}[Conservativity]\label{thm:conservativity}When $\chcon$ is symmetric and transitive we have

  \begin{mathpar}
        \bisim_o\;=\;\bisim{}
    \and \apitransarrow{}_o\;=\;\apitransarrow{}
    \and \wbisim_o\;=\;\wbisim
    \and \underset{{\rm tau}}{\wbisim}_o\;=\;\underset{{\rm tau}}{\wbisim}
    \and \wcong_o\;=\;\wcong
  \end{mathpar}
\end{thm}

\begin{proof}
  Formally proven in Isabelle.
  The bulk of the proof is to show that $\apitransarrow{}_o\;=\;\apitransarrow{}$.
  
  The $\Leftarrow$ direction is by induction on the derivation of the
  $\apitransarrow{}$ judgement, using symmetry to reorient the connectivity judgement
  in the \textsc{In} case. In the \textsc{Com} case, associativity and
  Lemma~\ref{lemma:provchaneq} are used to reconstruct the missing channel equivalence judgement

  The $\Rightarrow$ direction is by induction on the $\apitransarrow{}_o$ judgement, and is more
  involved; in particular, the \textsc{Com-Old} case requires relabelling the transitions
  obtained from the induction hypotheses with the provenance of the other.
\end{proof}

\noindent Our reduction semantics departs from standard designs~\cite{Berry:1989:CAM:96709.96717,Milner:1990:FP:90397.90426}
by relying on reduction contexts~\cite{DBLP:journals/tcs/FelleisenH92} instead of
structural rules, for two reasons. First, standard formulations tend to include rules like these:
\begin{mathpar}
\inferrule*[]
    {\trans{P}{}{P'}}   
{\trans{P \parop Q}{}{P' \parop Q}}
\and
\inferrule*[]
    {\ }
{\trans{\alpha.P + Q \parop \overline{\alpha}.R + S}{}{P \parop R}}
\end{mathpar}
\noindent A parallel rule like the above would be unsound because $Q$ might contain assertions that retract some conditions needed to derive $P$'s reduction.
The reduction axiom assumes prefix-guarded choice. We want our semantics to
apply to the full calculus, without limiting the syntax
to prefix-guarded $\caseonly$ statements.

But first, a few auxiliary definitions.
The \emph{reduction contexts} are the contexts in which communicating
processes may occur:

\begin{defi}[Reduction contexts]
The \emph{reduction contexts}
, ranged over by $C$, are generated by the grammar
\[
\begin{array}{rrll}
 C & := & P_G & \mbox{(process)} \\
  &  & \chole & \mbox{(hole)} \\
  &  & C \parop C & \mbox{(parallel)} \\
  &  & \caseonly\;{\ci{\ve{\varphi}}{\ve{P_G}}}\casesep{\ci{\varphi'}{C}}\casesep{\ci{\ve{\varphi''}}{\ve{Q_G}}} & \mbox{(case)} \\ 
\end{array}
\]
Let $\holes{C}$ denote the number of holes in $C$. $\cfill{C}{\ve{P_G}}$ denotes the process that results from filling each hole of $C$ with the corresponding element of $\ve{P_G}$, where holes are numbered from left to right; if $\holes{C} \neq |\ve{P_G}|$, $\cfill{C}{\ve{P_G}}$ is undefined.
\end{defi}

We do not need restriction contexts---instead, we rely on structural rules to
pull all restrictions to the top level.
To this end, we let \emph{structural congruence} $\equiv$ be the smallest equivalence relation on processes derivable using Theorems~\ref{thm:bisimpres} and~\ref{thm:strong-struct}.
The \emph{conditions} $\conds{C}$ and \emph{parallel processes} $\ppr{C}$ of a context $C$
are, respectively, the conditions in $C$ that guard the holes, and the processes of $C$ that
are parallel to the holes:

\[
\begin{array}{rrl}
 \ppr{P_G} & = & P_G \\
 \ppr{\chole} &  = & \nil \\
 \ppr{C_1 \parop C_2} &  = & \ppr{C_1} \parop \ppr{C_2} \\
 \ppr{\caseonly\;{\ci{\ve{\varphi}}{\ve{P_G}}}\casesep{\ci{\varphi'}{C}}\casesep{\ci{\ve{\varphi''}}{\ve{Q_G}}}} &  = & \ppr{C} \\
 & & \\
 \conds{P_G} &  = & \emptyset \\
 \conds{\chole} &  = & \emptyset \\
 \conds{C_1 \parop C_2} &  = & \conds{C_1} \cup \conds{C_2} \\
 \conds{\caseonly\;{\ci{\ve{\varphi}}{\ve{P_G}}}\casesep{\ci{\varphi'}{C}}\casesep{\ci{\ve{\varphi''}}{\ve{Q_G}}}} &  = & \{\varphi'\} \cup \conds{C}
\end{array}
\]

\begin{defi}[Reduction semantics]
The reduction relation $\mathord{\apitransarrow{}}\subseteq\processes \times \processes$ is
defined inductively by the rules of Table~\ref{table:reduction-semantics}.
\end{defi}
  
\begin{table}[tb]
\begin{mathpar}
\inferrule*[left=\textsc{Struct}]
    {P \equiv Q \\ \trans{Q}{}{Q'} \\ Q' \equiv P'}
    {\trans{P}{}{P'}}

\inferrule*[left=\textsc{Scope}]
    {\trans{P}{}{Q}}
    {\trans{(\nu a)P}{}{(\nu a)Q}}

\inferrule*[left=\textsc{Ctxt}]
    {\ve{\Psi} \vdash M \chcon N \\
      K = L\subst{\ve x}{\ve T} \\
      \forall \varphi \in \conds{C}.\;\ve{\Psi} \vdash \varphi
     }
    {\trans{\ve{\pass{\Psi}} \parop \cfill{C}{\outprefix{M}{K}.P,\;\inprefix{N}{\ve x}{L}.Q}}{}{\ve{\pass{\Psi}} \parop P \parop Q\subst{\ve x}{\ve T} \parop \ppr{C}}}
\end{mathpar}
\caption{\rm Reduction semantics.
  Here $\ve{\Psi}$ abbreviates the composition $\Psi_1 \otimes \Psi_2 \otimes \dots$, and
  $\ve{\pass{\Psi}}$ abbreviates the parallel composition $\pass{\Psi_1} \parop \pass{\Psi_2} \parop \dots$---for empty sequences they are taken to be $\one$ and $\nil$ respectively.
}
\label{table:reduction-semantics}
\end{table}

In words, \textsc{Ctxt} states that if an input and output prefix occur
in a reduction context, we may derive a reduction if the following holds:
the prefixes are connected in the current assertion environment,
the message matches the input pattern, and all conditions guarding
the prefixes are entailed by the environment.
The $\ppr{C}$ in the reduct makes sure any processes in parallel to the holes
are preserved.

Note that even though an unrestricted parallel rule would be unsound in a psi-calculus with
non-monotonic composition, the following is valid as a derived rule:
\[
\inferrule*[]
    {\trans{P}{}{P'}}   
{\trans{P \parop Q_G}{}{P' \parop Q_G}}
\]

\begin{thm}\label{lemma:harmony}
  $\trans{P}{}{P'}$ iff there is $P''$ such that $\framedtrans{\unit}{P}{\tau}{P''}$ and $P'' \equiv P'$
\end{thm}

\begin{proof}
  A full proof is available in the technical report~\cite{pohjola:newpsireport}.
  The $\Leftarrow$ direction is by induction on the derivation of $\trans{P}{}{P'}$.
  The $\Rightarrow$ direction is via reduction to normal form, showing that for every
  process $P$ there are $\ve{x},\Psi,P_G$ such that
  \[P \equiv (\nu \ve{x})(\ve{\pass{\Psi}} \parop P_G)\tag*{\qedhere}\]
\end{proof}

For barbed bisimulation, we need to define what the observables are, and what contexts
an observer may use.
We follow previous work by Johansson et al.~\cite{DBLP:conf/lics/JohanssonBPV10}
on weak barbed bisimilarity for the original psi-calculi on both counts.
First, we take the barbs to be the output labels a process can exhibit: we define
$P\exposes{\outlabel{M}{(\nu \ve{a})N}}$
($P$ exposes $\outlabel{M}{(\nu \ve{a})N}$)
to mean
$\exists P'.\;\framedtrans{1}{P}{\outlabel{M}{(\nu \ve{a})N}}{P'}$.
We write $P \exposes{\overline{M}}$ for
$\exists \ve{a},N. P \exposes{\outlabel{M}{(\nu \ve{a})N}}$,
and $P\wexposes{\alpha}$ for $P \goesto{\tau}^\ast\exposes{\alpha}$.
Second, we let observers use \emph{static} contexts, i.e.~ones built from parallel
and restriction.

\begin{defi}[Barbed bisimilarity]\label{def:barbbisim} \emph{Barbed bisimilarity}, written $\barbbisim$, is the largest equivalence on processes such that $P \barbbisim Q$ implies
\begin{enumerate}
  \item If $P\exposes{\outlabel{M}{(\nu \ve{a})N}}$
    and $\ve{a} \freshin Q$ then $Q\exposes{\outlabel{M}{(\nu \ve{a})N}}$ (barb similarity)
  \item If $\trans{P}{}{P'}$ then there exists $Q'$ such that $\trans{Q}{}{Q'}$ and $P' \barbbisim Q'$ (reduction simulation)
  \item $(\nu\ve{a})(P \parop R) \barbbisim (\nu\ve{a})(Q \parop R)$ (closure under static contexts)
\end{enumerate}
\end{defi}

Our proof that barbed and labelled bisimilarity coincides only
considers psi-calculi with a certain minimum of sanity and expressiveness.
This rules out some degenerate cases:
psi-calculi where there are messages that can be sent but not received,%
\footnote{
  More precisely, we refer to messages that can occur as objects of output transitions,
  but not as the objects of input transitions.
  This can happen in psi-calculi where the substitution function on terms
  is not surjective, because
  the rule \textsc{In} requires the message to be the result of a substitution.},
and psi-calculi where no transitions whatsoever are possible.

\begin{defi} A psi-calculus is \emph{observational} if:
\begin{enumerate}
  \item For all $P$ there are $M_P,K_P$ such that
    $\frameof{P} \vdash M_P \chcon K_P$ and
    not $P\wexposes{\overline{K_p}}$.
 \item If $N = (\ve{x}\;\ve{y})\cdot M$ and $\ve{y} \freshin M$ and
   $\ve{x},\ve{y}$ are distinct
   then $M\subst{\ve{x}}{\ve{y}} = N$.
\end{enumerate}
\end{defi}

\noindent The first clause means that no process can exhaust the set of barbs.
Hence observing contexts can signal success or failure without interference
from the process under observation.
For example, in the pi-calculus $M_P,K_P$ can be any name $x$ such that $x\freshin P$.
%
%
The second clause states that for swapping of distinct names,
substitution and permutation have the same behaviour.
Any standard definition of simultaneous substitution should satisfy this requirement.
These assumptions are present, explicitly or implicitly,
in the work of Johansson et al.~\cite{DBLP:conf/lics/JohanssonBPV10}.
Ours are given a slightly weaker formulation.

We can now state the main result of this section:


\begin{thm} In all observational psi-calculi, $P \barbbisim Q$ iff $\trisimsub{\one}{P}{Q}$.
\end{thm}

\begin{proof}
  A full proof is available in the technical report~\cite{pohjola:newpsireport}.
  Soundness is by coinduction on the definition of barbed bisimilarity,
  using $\bisim_{\one}$ as candidate relation,
  and using Theorems~\ref{lemma:harmony} and \ref{thm:bisimpres}, and the
  fact that $\equiv$ is a strong bisimulation.
  Completeness is by showing that $\{(\Psi,P,Q) : P \parop \pass{\Psi} \barbbisim Q \parop \pass{\Psi}\}$ is a bisimulation relation.
\end{proof}


\section{Expressiveness}

In this section, we study two examples of
the expressiveness gained by dropping symmetry and transitivity.

\subsection{Pi-calculus with preorders}\label{sec:prepi}

Recall that pi-F~\cite{gardner.wischik:explicit-fusions}
extends the pi-calculus with name equalities $(x=y)$ as first-class processes.
Communication in pi-F gives rise to equalities rather than substitutions,
so e.g.~$xy.P \parop \overline{x}z.Q$ reduces to $y = z \parop P \parop Q$:
the input and output objects are fused.
Hirschkoff et al.~\cite{DBLP:conf/lics/HirschkoffMS13} observe that fusion and
subtyping are fundamentally incompatible, and
propose a generalisation of pi-F
called the \emph{pi-calculus with preorders} or ${\pi}\!P$
to resolve the issue.

We are interested in ${\pi}\!P$ because its channel connectivity
is not transitive.
The equalities of pi-F are replaced with \emph{arcs} $a/b$ (``$a$ is above $b$'')
which act as one-way fusions: anything that can be done with $b$ can be done with $a$,
but not the other way around.
The effect of a communication is to create an arc
with the output subject above the input subject,
so $x(y).P \parop \overline{x}(z).Q$ reduces to $(\nu yz)(z/y \parop P \parop Q)$.
We write $A \vdash x \prec y$ to mean that $x$ and $y$ are related by the reflexive and transitive closure of the set of arcs $A$. $A$ is usually the set
of top-level arcs in the process under consideration, and will often be left implicit.
Two names $x,y$ are considered \emph{joinable} for the purposes of synchronisation if some name
$z$ is above both of them: formally, we write $x \curlyvee y$ for
$\exists z. x \prec z \wedge y \prec z$.

Hirschkoff et al.~conclude by saying that ``[it] could also be interesting to study the
representation of $\pi\!P$ into Psi-calculi.
This may not be immediate because the latter make use of on an equivalence relation on channels,
while the former uses a preorder''~\cite[p.~387]{DBLP:conf/lics/HirschkoffMS13}.
Having lifted the constraint that channels form an equivalence relation,
we happily accept the challenge.
We write ${\Psi}\!P$ for the psi-calculus we use to embed ${\pi}\!P$.
We follow the presentation of ${\pi}\!P$ from~\cite{DBLP:journals/jlp/HirschkoffMX15,DBLP:conf/fsen/HirschkoffMX15},
where the behavioural theory is most developed.

\begin{defi}
The psi-calculus ${\Psi}\!P$ is defined with the following parameters:
\[
\begin{array}{rrl}

  \terms & \defn & \nameset
  \\
  \conditions & \defn &
  \{x \prec y : x,y \in \nameset \}
  \cup
  \{x \curlyvee y : x,y \in \nameset \}

  \\
  \assertions & \defn & \powerfin{\{x \prec y : x,y \in \nameset \}}

  \\
  \unit & \defn & \{\}

  \\

  \otimes & \defn & \cup

  \\

  \chcon & \defn & \curlyvee

  \\

  \vdash & \defn & \mbox{the relation denoted $\vdash$ in~\cite{DBLP:conf/fsen/HirschkoffMX15}}.
  
\end{array}
\]
%
%
%
%
%
%
%
%
%
%
\end{defi}

\noindent The prefix operators of ${\pi}\!P$ are different from those of psi-calculi:
objects are always bound, communication gives rise to an arc rather than a
substitution, and a conditional silent prefix $[\varphi]\tau.P$ is included.
The full syntax, ignoring protected prefixes, is as follows:%
\footnote{We ignore protected prefixes because they are redundant, cf.~Remark~1 of \cite{DBLP:journals/jlp/HirschkoffMX15}.}

\begin{defi}[Syntax of ${\pi}\!P$]$\,$
  
  \begin{center}
  \begin{tabular}{rcll}
    $P$ & := & $a/b$ & (arc) \\
      &  & $\Sigma_{i\in I}\pi_i.P_i$ & (prefix-guarded choice) \\
      &  & $P \parop Q$ & (parallel) \\
      &  & $(\nu x)P$ & (restriction) \\
    $\pi$  & :=  & $a(x)$ & (input) \\
      &  & $\overline{a}{(y)}$ & (output) \\
      &  & $[\varphi]\tau$ & (conditional silent prefix) \\
    $\varphi$  & :=  & $x \prec y$ & \\
      &  & $x \curlyvee y$ &
  \end{tabular}
  \end{center}

Here $I$ is a finite index set.
\end{defi}

We can now define our encoding of ${\pi}\!P$ prefixes:

\begin{defi}[Encoding of prefixes]\label{def:prefixenc}
The encoding $\semb{\_}$ from ${\pi}\!P$ to ${\Psi}\!P$ is homomorphic on all
operators except prefixes and arcs, where it is defined by
\[
\begin{array}{rrll}
  \semb{a/b} & = & \pass{b \prec a} &
  \\
  \semb{\overline{a}{(y)}.P} & = & (\nu xy)(\overline{a}{x}.(\pass{x \prec y} \parop \semb{P}) &\mbox{ where $x\freshin y,P$}
  \\
  \semb{{a}{(y)}.P} & = & (\nu y)(\inprefix{a}{x}{x}.(\pass{y \prec x} \parop \semb{P})) &\mbox{ where $x\freshin y,P$}
  \\
  \semb{[\varphi]\tau.P} & = & \caseonly\;{\ci{\varphi}{(\nu x)(\inprefix{x}{x}{x}.0 \parop \overline{x}{x}.\semb{P})}} & \mbox{ where $x\freshin P$}
\end{array}
\]

For choice, we let $\semb{\Sigma_{i\in I}P_i} = {\caseonly\;{\ci{\ve{\varphi}}{\ve{\semb{P}}}}}$, where each $\varphi_i$ is a condition that is always entailed.\footnote{
  Such a condition can either be added to the target language, or we can use e.g.~%
  $a \prec a$ at the cost of some technical inconvenience.
  See the technical report for details.
  }
\end{defi}

\noindent This embedding of ${\pi}\!P$ in psi-calculi comes with a notion of bisimilarity per
Definition~\ref{def:strongbisim}.
We show that it coincides with the labelled bisimilarity for ${\pi}\!P$ (written $\sim$)
introduced in~\cite{DBLP:journals/jlp/HirschkoffMX15,DBLP:conf/fsen/HirschkoffMX15}.

\begin{thm}\label{thm:madiotbisim}
  $P \sim Q$ iff $\semb{P} \bisim \semb{Q}$
\end{thm}

\begin{proof}
  A full proof is available in the technical report~\cite{pohjola:newpsireport}.  
  We prove strong operational correspondence by a tedious induction,
  then prove (respectively) that
  $\{(P,Q). \trisimsub{\one}{\semb{P}}{\semb{Q}}\}$ is a bisimulation,
  and that
  $\{(\Psi,\semb{P},\semb{Q}). P \parop \Psi \sim Q \parop \Psi \}$ is
  a bisimulation up to $\bisim$.
\end{proof}

Thus our encoding validates the behavioural theory of ${\pi}\!P$ by connecting it to our fully mechanised proofs, while also showing that a substantially different design of the LTS
yields the same bisimilarity.
We will briefly compare these designs.
While we do rewriting of subjects in the prefix rules,
Hirschkoff et al. instead use relabelling rules like this one
(mildly edited to match our notation):
\begin{mathpar}
\inferrule*[]
    {\trans{P}{a(x)}{P'} \\ \frameof{P} \vdash a \prec b}
    {\trans{P}{b(x)}{P'}}
\end{mathpar}
\noindent An advantage of this rule is that it allows input and output labels
to be as simple as pi-calculus labels.
A comparative disadvantage is that it is not syntax-directed, and that the LTS has
more rules in total.
Note that this rule would not be a viable alternative to provenances in psi-calculi:
since it can be applied more than once in a derivation,
its inclusion assumes that the channels form a preorder wrt.~connectivity.

${\pi}\!P$ also has labels $[\varphi]\tau$, meaning that a silent
transition is allowed in environments where $\varphi$ is true.
A rule for rewriting $\varphi$ to a weaker condition, similar to the above rule for
subject rewriting, is included.
Psi-calculi does not need this because the \textsc{Par} rules take the
assertion environment into account.
${\pi}\!P$ transitions of kind $\trans{P}{[\varphi]\tau}{P'}$
correspond to $\Psi\!P$ transitions of kind $\framedtrans{\{\varphi\}}{P}{\tau}{P'}$.

Interestingly, the analogous full abstraction result fails to hold for the
embedding of pi-F in psi-calculi by Bengtson et al.~\cite{bengtson.johansson.ea:psi-calculi-long},
because outputs that emit distinct but fused names are distinguished by psi-calculus bisimilarity.
This issue does not arise here because $\pi\!P$ objects are always bound; however, we believe the
encoding of Bengtson et al.~can be made fully abstract by encoding free output with bound output,
exploiting the pi-F law $a\,y.Q \sim a(x)(Q \parop x=y)$.%


\subsection{Mixed choice}\label{sec:choice}
This section will argue that because we allow non-transitive channel connectivity,
the $\caseonly$ operator of psi-calculi becomes superfluous.
The formal results here will focus on encoding the special case of mixed choice.
We will then briefly discuss how to generalise these results
to the full $\caseonly$ operator.

Choice, written $P + Q$, is a process that behaves as either $P$ or $Q$.
In psi-calculi we consider $P + Q$ to abbreviate
${\caseonly\;{\ci{\top}{P}}\casesep{\ci{\top}{Q}}}$
for some condition $\top$ that is always entailed.
\emph{Mixed choice} means that in $P + Q$,
$P$ and $Q$ must be prefix-guarded.
In particular, mixed choice allows choice between an input and an output.
There is a straightforward generalisation to $n$-ary sums that, in order
to simplify the presentation, we will not consider here.



Fix a psi-calculus
$\mathcal{P}=(\terms,\assertions,\conditions,\vdash,\otimes,\unit,\chcon)$
with mixed choice. This will be our source language.
For technical convenience we assume that $\mathcal{P}$ satisfies
the equation $\one\sigma = \one$ for all substitutions $\sigma$;
see the associated technical report~\cite{pohjola:newpsireport}
for a discussion on how this assumption can be lifted.
We will construct a target psi-calculus and an encoding such that the
target terms make no use of the $\caseonly$ operator.
The target language $\mathcal{E}(\mathcal{P})$ adds to $\terms$ the
ability to tag a term $M$ with a name $x$; we write $M_x$ for the
tagged term.
We write $\alpha_x$ for tagging the subject of the prefix $\alpha$ with $x$.
Tags are used to uniquely identify which choice statement a prefix is a summand of.
As the assertions of $\mathcal{E}(\mathcal{P})$ we use $\assertions \times \powerfin{\nameset}$,
where $\powerfin{\nameset}$ are the \emph{disabled tags}.

\begin{defi}[Target language]\label{def:targetlang} Let
  $\mathcal{P}=(\terms,\assertions,\conditions,\vdash,\otimes,\unit,\chcon)$
be a psi-calculus.
Then $\mathcal{E}(\mathcal{P}) = (\terms_\mathcal{E},\assertions_\mathcal{E},\conditions_\mathcal{E},\vdash_\mathcal{E},\otimes_\mathcal{E},\unit_\mathcal{E},\chcon_\mathcal{E})$
is a psi-calculus whose components are as follows:

\[
\begin{array}{rlll}
  \terms_\mathcal{E} & = & \multicolumn{2}{l}{\terms \uplus \{M_x: x \in \nameset, M \in \terms_\mathcal{E}\}}
  \\
  \assertions_\mathcal{E} & = & \multicolumn{2}{l}{\assertions \times \powerfin{\nameset}
  }\\
  \conditions_\mathcal{E} & = & \multicolumn{2}{l}{\conditions \uplus \{M \chcon N: M,N \in \terms\} \uplus \nameset
  }\\
  (\Psi,\mathbf{N}) \otimes_\mathcal{E} (\Psi',\mathbf{N}') & = & \multicolumn{2}{l}{(\Psi \otimes \Psi',\mathbf{N} \cup \mathbf{N}')
  }\\
  \unit_\mathcal{E} & = & (\one,\emptyset)
  &\\
  & & & \\
  (\Psi,\mathbf{N}) & \vdash_\mathcal{E} & \varphi & \mbox{if } \varphi \in \conditions \mbox{ and } \Psi \vdash \varphi
  \\
  (\Psi,\mathbf{N}) & \vdash_\mathcal{E} & x & \mbox{if } x \in \nameset \mbox{ and } x \in \mathbf{N}
  \\
  (\Psi,\mathbf{N}) & \vdash_\mathcal{E} & M_x \chcon N_y & \mbox{if }  \Psi \vdash M \chcon N \mbox{ and } x \neq y \mbox{ and } x,y\notin\mathbf{N}
  \\
  (\Psi,\mathbf{N}) & \vdash_\mathcal{E} & M_x \chcon N & \mbox{if }  \Psi \vdash M \chcon N \mbox{ and } x\notin\mathbf{N}
  \\
  (\Psi,\mathbf{N}) & \vdash_\mathcal{E} & M \chcon N_x & \mbox{if }  \Psi \vdash M \chcon N \mbox{ and } x\notin\mathbf{N}
\end{array}
\]
%
%
\end{defi}

We assume that the target language is constrained by a sorting system as in~\cite{DBLP:conf/tgc/BorgstromGPVP13}
that ensures only terms $M \in \terms$ can ever occur as objects of communication,
and in particular, that for all substitutions $\subst{\ve{x}}{\ve{T}}$,
$\ve{T} \subseteq \terms$.
The purpose of this simplification is to avoid having to consider input transitions such as
\[\framedtrans{\Psi}{\inprefix{M}{\ve{x}}N.P}{\inlabel{K}{K_x}}{P'}\]
\noindent that may result in substitutions where tagged terms must be substituted into
source-language terms or vice versa. It is possible to lift this assumption at the cost
of significant technical inconvenience~\cite{pohjola:newpsireport}.

The encoding $\semb{\_}$ from $\mathcal{P}$ to $\mathcal{E}(\mathcal{P})$
is homomorphic on all operators except assertion and choice, where it is defined as follows:
\begin{mathpar}
  \semb{\pass{\Psi}} = \pass{(\Psi,\emptyset)}
  \and
  \semb{\alpha.P + \beta.Q} =
  (\nu x)(\alpha_x.(\semb{P} \parop \pass{(\unit,\{x\})}) \parop \beta_x.(\semb{Q} \parop \pass{(\unit,\{x\})}))
\end{mathpar}
\noindent where $x \freshin \alpha,\beta,P,Q$.
\noindent If we disregard the tag $x$, we see that the encoding simply offers up both summands
in parallel. This clearly allows all behaviours of $\alpha.P + \beta.Q$, but there are two
additional behaviours we must prevent: (1) communication between the summands, and (2)
lingering summands firing after the other branch has already been taken.
The tagging mechanism prevents both,
as a consequence of how we define channel equivalence on tagged terms in $\mathcal{E}(\mathcal{P})$: tagged channels are connected if the underlying channel is connected.
To prevent (1), Definition~\ref{def:targetlang} requires the tags of connected channels
to be different,
and to prevent (2) the definition requires that the tags
are not disabled. Note that this channel connectivity is not transitive, not reflexive, and
not monotonic wrt.~assertion composition---not even if the source language connectivity is.

\begin{exa}
  We illustrate the operational behaviour of the encoding for the process $R = \alpha.P + \beta.Q$.
  Its encoding is
  \[\semb{R} = (\nu x)(\alpha_x.(\semb{P} \parop \pass{(1,\{x\})})  \parop \beta_x.(\semb{Q} \parop \pass{(1,\{x\})}))\]
  \noindent where $x$ is a fresh name. Suppose $\alpha$ is an output prefix with subject $M$,
  and that channel connectivity is reflexive.
  Then we can derive the transition $\trans{R}{\alpha}{P}$.
  The corresponding derivation from $\semb{R}$ uses the connectivity judgement $M_x \chcon M$ in the \textsc{Out} rule to derive the following transition:
  \[\trans{\semb{R}}{\alpha}{(\nu x)(\semb{P} \parop S)}\]
  \noindent where
  \[S = \pass{(1,\{x\})} \parop b_x.(\semb{Q} \parop \pass{(1,\{x\})})\]
  \noindent Since $x$ is fresh in $\semb{P}$, by scope extension we have
  $(\nu x)(\semb{P} \parop S) \bisim \semb{P} \parop (\nu x)S$.
  Moreover, we have $(\nu x)S \bisim 0$. To see why, note first that $(\nu x)S$ has no
  outgoing transitions: its only prefix has the tag $x$, which is disabled by its
  top-level assertion $\pass{(1,\{x\})}$. Second, note that since this disabled tag is a local
  name, its disabling has no effect on the environment.
\end{exa}

\begin{thm}[Correctness of choice encoding]\label{thm:choicecorrect}\ 
  \begin{enumerate}
  \item If $\framedtrans{\Psi}{P}{\alpha}{P'}$
    then there is $P''$ such that $\framedtrans{(\Psi,\emptyset)}{\semb{P}}{\alpha}{P''}$
    and $\trisimsub{(\Psi,\emptyset)}{P''}{\semb{P'}}$.
  \item If $\framedtrans{(\Psi,\emptyset)}{\semb{P}}{\alpha}{P'}$
    then there is $P''$ such that $\framedtrans{\Psi}{P}{\alpha_\bot}{P''}$
    and $\trisimsub{(\Psi,\emptyset)}{P'}{\semb{P''}}$.
  \item $\trisimsub{\one}{P}{Q}$ iff $\trisimsub{(\one,\emptyset)}{\semb{P}}{\semb{Q}}$.
  \end{enumerate}
\end{thm}

\begin{proof}
  A full proof is available in the technical report~\cite{pohjola:newpsireport}.
  Forward simulation is by induction on the derivation of $\framedtrans{\Psi}{P}{\alpha}{P'}$,
  backward simulation by structural induction on $P$ followed by
  inversion on the derivation of the transition from $\semb{P}$.
  Full abstraction is by showing that
    \[\{((\Psi,\mathbf{N}),\semb{P},\semb{Q}) : \trisimsub{(\Psi}{P}{Q}\}\]
  and
    \[\{(\Psi,P,Q) : \trisimsub{(\Psi,\emptyset)}{\semb{P}}{\semb{Q}}\}\]
  are bisimulation relations.
\end{proof}

\noindent Here $\alpha_\bot$ denotes the label $\alpha$ with all tags removed.
It is immediate from Theorem~\ref{thm:choicecorrect} and the definition of $\semb{\_}$
that our encoding also satisfies the other standard quality criteria~\cite{DBLP:journals/iandc/Gorla10}:
it is compositional, it is name invariant, and it preserves and reflects barbs and divergence.

In the original psi-calculi, our target language is invalid because of its non-transitive
connectivity.
However,
if we restrict attention to \emph{separate} choice (where either both summands are inputs
or both summands are outputs),
a slight modification of the scheme above yields a correct encoding
in the context of the original psi-calculi.
With separate choice we can drop the side-condition that tags of connected processes
are distinct from Definition~\ref{def:targetlang}---
the only use of this side-condition is to prevent communication between summands,
and separate choice already prevents this by construction.
With this modified definition of $\chcon_{\mathcal{E}}$,
we have that if $\chcon$ is symmetric and transitive, then
so is $\chcon_{\mathcal{E}}$. Or in other words, if the source language
is expressible in the original psi-calculi then so is the target language.

These results generalise in a straightforward way to mixed $\caseonly$
statements
\[{\caseonly\;{\ci{\varphi_1}{\alpha.P}}\casesep{\ci{\varphi_2}{\beta.Q}}}\]
by additionally tagging terms with a condition, i.e.~$M_{x,\varphi_1}$, that must be entailed
in order to derive connectivity judgements involving the term.
The generalisation to free choice, i.e.~$P+Q$ where $P,Q$ can be anything,
is more involved and sacrifices some compositionality.
The idea is to use sequences of tags, representing which branches
of which (possibly nested) case statements a prefix can be found in, and disallowing
communication between prefixes in distinct branches of the same $\caseonly$ operator.

\section{Conclusion and related work}

We have seen how psi-calculi can be conservatively extended to allow asymmetric and non-transitive
communication topologies, sacrificing none of the bisimulation meta-theory.
This confers enough expressiveness to capture a pi-calculus with preorders,
and makes mixed choice a derived operator.

The work of Hirschkoff et al.~\cite{DBLP:conf/lics/HirschkoffMS13}
is closely related in that it uses non-transitive connectivity;
see Section~\ref{sec:prepi} for an extensive discussion.

Broadcast psi-calculi~\cite{DBLP:journals/sosym/BorgstromHJRVPP15} extend
psi-calculi with broadcast communication in addition to point-to-point communication.
There, point-to-point channels must still be symmetric and transitive,
but for broadcast channels this condition is lifted, at the cost
of introducing other side-conditions on how names are used:
broadcast prefixes must be connected via intermediate \emph{broadcast channels}
which have no greater support than either of the prefixes it connects,
precluding language features such as name fusion.
We believe provenances could be used to define a version of broadcast psi-calculi
that does not need this side-condition.

Kouzapas et al.~\cite{DBLP:journals/corr/KouzapasGG14} define a similar
reduction context semantics for (broadcast) psi-calculi.
Their reduction contexts requires three kinds of numbered holes with
complicated side-conditions on how the holes may be filled;
we have attempted to simplify the presentation by having
only one kind of hole.
While (weak) barbed congruence for psi-calculi has been studied before
\cite{DBLP:conf/lics/JohanssonBPV10} (see Section~\ref{sec:validation}),
barbed congruence was defined in terms of the labelled semantics rather
than a reduction semantics, thus weakening its claim to independent
confirmation slightly.

There is a rich literature on choice encodings for the pi-calculus%
~\cite{DBLP:journals/iandc/Gorla10,DBLP:journals/iandc/NestmannP00,DBLP:conf/popl/Palamidessi97,DBLP:conf/fossacs/PetersN12,DBLP:conf/esop/PetersNG13},
with many separation and encodability results
under different quality criteria for different flavours of choice.
Encodings typically require complicated protocols and tradeoffs between quality criteria.
Thanks to the greater expressive power of psi-calculi,
our encoding is simpler and satisfies stronger quality criteria
than any choice encoding for the pi-calculus.
Closest to ours is the choice encoding of CCS into
the DiX calculus by Busi and Gorrieri~\cite{DBLP:conf/ecoopw/BusiG94}.
DiX introduces a primitive for annotating processes with \emph{conflict sets},
that are intended as a generalisation of choice.
Processes with overlapping conflict sets cannot interact, and when a process acts, every
process with an overlapping conflict set is killed.
These conflict sets perform the same role in the encoding as our tags do.
We believe the tagging scheme used in our choice encoding also captures DiX-style
conflict sets.

\section*{Acknowledgements}
These ideas have benefited from discussions with many people at Uppsala University, ITU Copenhagen,
the University of Oslo and Data61/CSIRO, including
Jesper Bengtson, Christian Johansen, Magnus Johansson and Joachim Parrow.
I would also like to thank Jean-Marie Madiot and the anonymous reviewers
for valuable comments on earlier versions of the paper.

\FloatBarrier

\bibliographystyle{alpha}
\bibliography{bibliography}

\end{document}